\documentclass[journal,10pt,twoside,twocolumn]{IEEEtran}

\IEEEoverridecommandlockouts
 \usepackage{algorithm}
 \usepackage{algorithmic}
\usepackage{amsmath,amsfonts} %citesort removed
\usepackage{amssymb, amsthm}
\usepackage{graphicx}
\usepackage{subfigure}  % separates the subcaptions a bit from the figures
\usepackage{booktabs}
\usepackage{multirow}
\usepackage{cite}
\usepackage{upgreek}
\usepackage{color}
\usepackage{microtype}
\usepackage{balance}
\usepackage{hyperref}

\usepackage{microtype}
\usepackage{adjustbox}

\renewcommand{\eqref}[1]{(\ref{#1})}
\newcommand{\secref}[1]{\mbox{Section~\ref{#1}}}

\newcommand{\figref}[1]{\mbox{Figure~\ref{#1}}}
\newcommand{\tblref}[1]{\mbox{Table~\ref{#1}}}
\newcommand{\algref}[1]{\mbox{Algorithm~\ref{#1}}}

\renewcommand{\Pr}[1]{\ensuremath{\mathrm{P}\!\left[#1\right]}}

\DeclareMathOperator*{\argmin}{arg\;min}
\DeclareMathOperator*{\argmax}{arg\;max}

\newcommand{\setO}{\ensuremath{\mathcal{O}}}

\newcommand{\setX}{\ensuremath{\mathcal{X}}}

\newcommand{\bma}{\ensuremath{\mathbf{a}}}

\newcommand{\bmh}{\ensuremath{\mathbf{h}}}

\newcommand{\bml}{\ensuremath{\mathbf{l}}}

\newcommand{\bmn}{\ensuremath{\mathbf{n}}}

\newcommand{\bms}{\ensuremath{\mathbf{s}}}

\newcommand{\bmv}{\ensuremath{\mathbf{v}}}

\newcommand{\bmx}{\ensuremath{\mathbf{x}}}
\newcommand{\bmy}{\ensuremath{\mathbf{y}}}
\newcommand{\bmz}{\ensuremath{\mathbf{z}}}

\newcommand{\bA}{\ensuremath{\mathbf{A}}}

\newcommand{\bD}{\ensuremath{\mathbf{D}}}

\newcommand{\bH}{\ensuremath{\mathbf{H}}}
\newcommand{\bI}{\ensuremath{\mathbf{I}}}

\newcommand{\bL}{\ensuremath{\mathbf{L}}}

\newcommand{\bN}{\ensuremath{\mathbf{N}}}

\newcommand{\bS}{\ensuremath{\mathbf{S}}}
\newcommand{\bT}{\ensuremath{\mathbf{T}}}

\newcommand{\bV}{\ensuremath{\mathbf{V}}}

\newcommand{\bY}{\ensuremath{\mathbf{Y}}}

\newcommand{\bZero}{\ensuremath{\mathbf{0}}}
\newcommand{\bOne}{\ensuremath{\mathbf{1}}}

\newcommand{\MT}{{\ensuremath{U}}}
\newcommand{\MR}{{\ensuremath{B}}}

\newcommand{\No}{\ensuremath{N_{\mathrm{0}}}}

\newtheorem{thm}{Theorem}

\newtheorem{prop}[thm]{Proposition}
\makeatletter
\renewenvironment{proof}[1][\proofname]{\par
  \pushQED{\qed}%
  \normalfont \topsep6\p@\@plus6\p@\relax
  \trivlist
  \item[\hskip\labelsep
        \scshape
    #1\@addpunct{.}]\ignorespaces
}{%
  \popQED\endtrivlist\@endpefalse
}
\makeatother

\renewcommand{\bml}{\ensuremath{\boldsymbol \ell}}

\newcommand{\rev}[1]{\textcolor{black}{#1}}

\mathchardef\mhyphen="2D

\begin{document}
\title{Data Detection in Large Multi-Antenna Wireless \\ Systems via Approximate Semidefinite Relaxation}
\author{
\IEEEauthorblockN{Oscar Casta\~neda, Tom Goldstein, and Christoph Studer} 
\thanks{O. Casta\~neda and C. Studer are with the School~of ECE, Cornell University, Ithaca, NY; e-mail: oc66@cornell.edu, studer@cornell.edu} 
\thanks{T. Goldstein is with the Department of CS, University of Maryland, College Park, MD; e-mail: tomg@cs.umd.edu}
\thanks{A short version of this paper summarizing the TASER FPGA design for large MU-MIMO data detection has been presented at the  IEEE International Symposium on Circuits and Systems (ISCAS) 2016 \cite{CGS2016a}.}
\thanks{The system simulator for TASER used in this paper will be available on GitHub: \url{https://github.com/VIP-Group/TASER}}
}

\maketitle

\begin{abstract}
Practical data detectors for future wireless systems with hundreds of antennas at the base station must achieve high throughput and low error rate at low complexity. Since the complexity of maximum-likelihood (ML) data detection is prohibitive for such large wireless systems, approximate methods are necessary. 
\rev{In this paper, we propose a novel data detection algorithm referred to as Triangular Approximate SEmidefinite Relaxation (TASER), which is suitable for two application scenarios: (i) coherent data detection in large multi-user multiple-input multiple-output (MU-MIMO) wireless systems and (ii) joint channel estimation and data detection in large single-input multiple-output (SIMO) wireless systems.}
For both scenarios, we show that TASER achieves near-ML error-rate performance at low complexity by relaxing the associated ML-detection \rev{problems} into a semidefinite program, which we solve approximately using a preconditioned forward-backward splitting procedure.
Since the resulting problem is non-convex, we provide convergence guarantees for our algorithm. 
To demonstrate the efficacy of TASER in practice, we design a systolic architecture that enables our algorithm to achieve high throughput at low hardware complexity, \rev{and we develop reference field-programmable gate array (FPGA) and application-specific integrated circuit (ASIC) designs for various antenna configurations.}
\end{abstract}

%Practical data detectors for future wireless systems with hundreds of antennas at the base station must achieve high throughput and low error rate at low complexity. Since the complexity of maximum-likelihood (ML) data detection is prohibitive for such large wireless systems, approximate methods are necessary. In this paper, we propose a novel data detection algorithm referred to as Triangular Approximate SEmidefinite Relaxation (TASER), which is suitable for two application scenarios: (i) coherent data detection in large multi-user multiple-input multiple-output (MU-MIMO) wireless systems and (ii) joint channel estimation and data detection in large single-input multiple-output (SIMO) wireless systems. For both scenarios, we show that TASER achieves near-ML error-rate performance at low complexity by relaxing the associated ML-detection problem into a semidefinite program, which we solve approximately using a preconditioned forward-backward splitting procedure. Since the resulting problem is non-convex, we provide convergence guarantees for our algorithm. To demonstrate the efficacy of TASER in practice, we design a systolic architecture that enables our algorithm to achieve high throughput at low hardware complexity, and we develop reference field-programmable gate array (FPGA) and application-specific integrated circuit (ASIC) designs for various antenna configurations.
%

%
\begin{IEEEkeywords}
\rev{FPGA and ASIC design, data detection, joint channel estimation and data detection, large single-input and multiple-input multiple-output (SIMO and MIMO) wireless systems,  semidefinite relaxation.}
\end{IEEEkeywords}

%%%%%%%%%%%%%%%%%%%%%%%%%%%%%%%%%%%%%%%%%%%%%%%%%%%%

\section{Introduction}
\label{sec:intro}
\IEEEPARstart{L}{arge} \rev{multiple-input multiple-output (MIMO) and  single-input multiple-output (SIMO) wireless technology, where the base station (BS) is equipped with hundreds or thousands  of antennas,} are widely believed to play a major role in fifth-generation~(5G) cellular communication systems~\cite{Marzetta2010,Rusek2012,hoydis2013massive,larsson2014massive,andrews2014will,LUetal2014}.
Such large wireless systems promise improved spectral efficiency, coverage, and range compared to traditional small-scale systems.  However, the extremely large number of BS antennas requires the design of high-performance data-detection algorithms that can be implemented efficiently in very-large scale integration (VLSI) circuits~\cite{Wu2014}.
\rev{In fact, data detection is among the most critical baseband-processing tasks in terms of implementation complexity, power consumption, throughput, and error-rate performance for such systems~\cite{wong2002vlsi,burg2005vlsi}.}

\rev{To enable high-throughput uplink communication for massive multi-user (MU) MIMO wireless systems (where tens of user terminals transmit data to a BS with hundreds of antennas),} a variety of  low-complexity data-detection algorithms~\cite{vardhan2008,Prabhu2013,shengwu2014,svac2014,Hu2014,Yincg,LiuYGLS15,jeon2015optimality,li2015accelerating}, as well as a few field-programmable gate array (FPGA) implementations~\cite{Wu2014,Yin2015,Wu2016,WZXXY2016} and application-specific integrated circuit (ASIC) designs~\cite{YWWDCS14b} have been proposed recently.
To date, all data detectors that have been implemented in VLSI for such high-dimensional problems rely on (approximate) linear data detection~\cite{Wu2014,Yin2015,Wu2016,YWWDCS14b,WZXXY2016}.
Such linear methods are known to suffer from a significant error-rate performance loss for more realistic systems with a not-so-large number of antennas at the BS or where the number of user terminals is comparable to that of the number of BS antennas~\cite{Wu2014}.
Furthermore, the literature on large MU-MIMO data detection almost exclusively relies on the assumption of perfect channel state information (CSI) at the BS---an assumption that cannot be satisfied in practice.

\subsection{Contributions}
In this paper, we propose a novel data detection algorithm and corresponding VLSI designs for large wireless systems. 
\rev{Our algorithm, referred to as 
Triangular Approximate SEmidefinite Relaxation (TASER), 
can be deployed in two different application scenarios: (i) coherent data detection in massive MU-MIMO wireless systems and (ii) joint channel estimation and data detection (JED) in large  SIMO wireless systems.}
Our detector builds upon semidefinite relaxation~\cite{luo2010semidefinite}, which enables near maximum-likelihood (ML) data detection performance at polynomial (in the number of transmit antennas or time slots) complexity for systems that communicate with low-rate, constant-modulus modulation schemes~\cite{jalden2008diversity}. 
\rev{TASER approximates the semidefinite relaxation (SDR) formulation of both the coherent ML and the JED ML problems using a Cholesky factorization, and solves the resulting non-convex problem using a preconditioned forward-backward splitting (FBS) procedure~\cite{BT09,GSB14}.}
We provide theoretical convergence guarantees for our algorithm, and we develop a corresponding systolic array that enables high-throughput data detection at low silicon area in an energy-efficient manner.
We provide reference VLSI implementation results for a Xilinx Virtex-7 FPGA and for a 40\,nm CMOS technology, and we perform an extensive comparison in terms of performance and complexity with recently-proposed data detector implementations for large MU-MIMO wireless systems~\cite{Wu2014,Yin2015,Wu2016,WZXXY2016,YWWDCS14b}.

\subsection{Relevant Prior Art}
\subsubsection{Data detection in large MU-MIMO}
The literature on data detection in large (or massive) MU-MIMO wireless systems describes only a few algorithms that are able to achieve near-optimal error-rate performance \cite{vardhan2008,shengwu2014,jeon2015optimality}. For these algorithms, however, no hardware designs have been described in the open literature. So far, only sub-optimal, linear data detection algorithms have  been integrated  in FPGAs~\cite{Wu2014,Yin2015,Wu2016,WZXXY2016} or ASICs~\cite{YWWDCS14b}. Unfortunately, such linear data detection algorithms suffer from a significant error-rate performance loss in \rev{``square'' systems,} where the number of users is comparable to the number of BS antennas \cite{Wu2014}. In contrast, the proposed TASER algorithm achieves near-optimal error-rate  performance, even in symmetric large MU-MIMO systems where the BS-to-user-antenna ratio is one. 

\subsubsection{SDR-based data detection}
\rev{SDR is a well-known technique for achieving near-ML performance in multi-user code division multiple access (MU-CDMA) \cite{tan2001application,ma2002quasi} and traditional, small-scale MIMO \cite{steingrimsson2003soft,jalden2003semidefinite,wiesel2005semidefinite,sidiropoulos2006,yang2007mimo,ma2009equivalence,luo2010semidefinite,wai2011cheap} wireless systems.}
Most results on SDR-based data detection rely on computationally inefficient, general-purpose convex solvers that require either the  solution to a linear system or an eigenvalue decomposition per iteration---both of these operations entail prohibitive complexity when implemented in hardware.
As an exception, the algorithm in~\cite{wai2011cheap} relies on block-coordinate descent, which avoids the solution to a full linear system per iteration. While computationally efficient, this method exhibits stringent data dependencies, requires a high number of multiplications per iteration, and consumes a large amount of memory, which renders corresponding VLSI designs inefficient. 
TASER, in contrast, is highly parallelizable and hardware friendly, and is---to the best of our knowledge---the first SDR-based data detector that has been successfully implemented in VLSI. 

\subsubsection{Joint channel estimation and data detection}
JED is known to significantly outperform traditional data detection schemes that separate channel estimation from data detection.
We believe that JED is a promising solution for large cellular systems, where pilot-contamination  (i.e., pilot-based training for users in adjacent cells interferes with the training pilots in the current cell) poses a fundamental performance bottleneck~\cite{Marzetta2010}. 
\rev{The computational complexity of exact JED via an exhaustive search grows exponentially in the number of transmission time slots \cite{vikalo2006efficient}.} Hence, sphere-decoding (SD)-based methods have been proposed for JED in the SIMO \cite{stoica2003space,vikalo2006efficient,stoica2003joint,alshamary2015optimal} and MIMO \cite{xu2008exact} literature to reduce the computational complexity. Nevertheless, the design of hardware implementations of high-throughput sphere-decoders is challenging, and most existing designs  only achieve a few hundred Mb/s for small MIMO systems (see \cite{burg2005vlsi,studer2010vlsi} for more details on SD-based data detectors). In addition---to the best of our knowledge---no hardware design for JED has been proposed in the open literature. 
\rev{In this paper, we show that (i) JED can be performed using SDR and (ii) TASER enables near-optimal, high-throughput JED for realistic large SIMO wireless systems.}

\subsection{Notation}

Lowercase boldface letters stand for column vectors; uppercase boldface letters denote matrices. For a matrix~\bA, we denote its transpose, adjoint, and trace by $\bA^T$, $\bA^H$, and $\mathrm{Tr}(\bA)$, respectively.
We use~$A_{k,\ell}$ for the entry in the $k$th row and $\ell$th column of the matrix $\bA$; the $k$th entry of a vector $\bma$ is denoted by~$a_k=[\bma]_k$. The Frobenius norm of the matrix $\bA$ is $\|\bA\|_F=\sqrt{\sum_{k,\ell}|A_{k,\ell}|^2}$ and the $\ell_2$-norm of the vector~$\bma$ is $\|\bma\|_2=\sqrt{\sum_{k}|a_k|^2}$. 
The identity matrix and all-ones vector are denoted by $\bI$ and $\bOne$, respectively. 
The real and imaginary part of a complex-valued matrix $\bA$ are denoted by $\Re(\bA)$ and $\Im(\bA)$, respectively. 

\subsection{Paper Outline}
The rest of the paper is organized as follows. \secref{sec:system} introduces the large MU-MIMO and SIMO system models and discusses coherent ML  data detection as well as JED. \secref{sec:algo} introduces the TASER algorithm and provides a theoretical convergence analysis. \secref{sec:impl} details our systolic architecture. \secref{sec:results} shows reference implementation results and provides a comparison with existing data detectors for large MU-MIMO. \rev{Concluding remarks are presented in \secref{sec:conclusions}.}

%%%%%%%%%%%%%%%%%%%%%%%%%%%%%%%%%%%%%%%%%%%%%%%%%%%%

\section{Data Detection in Large Multi-Antenna Wireless Systems}
\label{sec:system}

The algorithm and VLSI \rev{designs} proposed in this paper are suitable for two application scenarios: (i) coherent data detection in large MU-MIMO systems and (ii) JED in large SIMO systems. 
We next describe the corresponding system models and show how both problems can be relaxed to a semidefinite program (SDP) of the same form.

\subsection{Coherent Data Detection for Large MU-MIMO Systems}

The first application scenario is data detection in the large (or massive) MU-MIMO wireless uplink with $\MR$ BS antennas and~$\MT$ user antennas. 
We consider the standard input-output relation to model a narrow-band\footnote{Our algorithm and circuit designs are also suitable for frequency-selective channels in combination with orthogonal frequency-division multiplexing (OFDM), where we consider the same input-output relation per subcarrier.} MIMO wireless channel~\cite{gesbert2003theory}: 
$\mathbf{y}=\mathbf{H}\mathbf{s}+\mathbf{n}$.
Here, $\mathbf{y}\in\mathbb{C}^\MR$ is the BS receive-vector, $\bH\in\mathbb{C}^{\MR\times\MT}$ is the MIMO channel matrix, $\bms\in\setO^{\MT}$ is the transmit vector containing the data symbols from all users ($\setO$ refers to the constellation set), and $\bmn\in\mathbb{C}^{\MR}$ is i.i.d.\ circularly-symmetric Gaussian with variance~$\No$ per entry. 
Assuming that an estimate of the channel matrix~$\bH$ was acquired during a dedicated training phase, \rev{ML} data detection corresponds to the following  problem \cite{Paulraj2008}:
\begin{align} \label{eq:MLproblem}
\hat{\bms}^\text{ML} = \argmin_{\bms\in\setO^U} \|\bmy-\bH\bms\|_2.
\end{align}
\rev{A number of computationally efficient sphere-decoding algorithms have been proposed to solve the combinatorial problem in~\eqref{eq:MLproblem} for conventional, small-scale MIMO systems~\cite{agrell2002closest,HB03,burg2005vlsi,jsac07}.} Unfortunately, the worst-case and average computational complexity of these exact methods still scales exponentially with the number of users~$\MT$~\cite{jalden2005complexity,seethaler2011complexity}. 
For large MU-MIMO systems, where the BS-to-user-antenna ratio \rev{exceeds a factor of two,} recently-developed linear algorithms have been shown to achieve near-ML performance~\cite{hoydis2013massive,larsson2014massive,Wu2014}. For systems with a large number of users where the BS-to-user-antenna ratio is close to one, however, linear methods are known to deliver  poor error-rate performance~\cite{Wu2014}.

\sloppy

\rev{To enable near-optimal error-rate performance at low complexity for such scenarios, we can relax the ML problem in \eqref{eq:MLproblem} into an SDP~\cite{luo2010semidefinite}.} This relaxation step requires us to reformulate the ML detection problem as follows. By assuming constant-modulus QAM constellations, such as BPSK and QPSK, we first perform the real-valued decomposition of the system model $\overline{\bmy}=\overline{\bH}\overline{\bms}+\overline{\bmn}$ using the following definitions:
\begin{align*}
& \overline{\bmy} \!=\! \!\left[\begin{array}{c}
 \!\! \Re(\bmy)  \!\! \\
 \!\! \Im(\bmy)   \!\! 
\end{array}\right]\!,  \quad 
& & \overline{\bH} \!=\!  \!\left[\begin{array}{cc}
 \!\! \Re(\bH)  \!\!\!\! &~~-\Im(\bH)  \!\! \\
 \!\! \Im(\bH)  \!\!\!\! &~~\Re(\bH)  \!\!
\end{array}\right]\!,  \\
& \overline{\bms} \!=\!  \!\left[\begin{array}{c}
 \!\! \Re(\bms)  \!\!  \\
 \!\! \Im(\bms)  \!\!  
\end{array}\right]\!,\quad  %\text{and} \quad
& & \overline{\bmn} \!=\! \!\left[\begin{array}{c}
 \!\! \Re(\bmn)  \!\!  \\
 \!\! \Im(\bmn)  \!\!  
\end{array}\right]\!.
\end{align*}
\rev{This decomposition enables us to reformulate the ML problem in~\eqref{eq:MLproblem} into the following equivalent form:}
\begin{align} \label{eq:reformulatedMLproblem}
\bar{\bms}^\text{ML}  =\argmin_{\tilde\bms\in\setX^{N}} \, \mathrm{Tr}(\tilde\bms^T\bT\tilde\bms).
\end{align}
For QPSK, the matrix $\bT=[\overline{\bH}^T\overline{\bH},-\overline{\bH}^T\overline{\bmy};-\overline{\bmy}^T\overline{\bH},\overline{\bmy}^T\overline{\bmy}]$ is of dimension $N\times N$ with $N=2\MT+1$ and $\setX\in\{-1,+1\}$ with $\tilde\bms=[\Re(\bms); \Im(\bms); 1]$.
The solution $\bar{\bms}^\text{ML} $ can then be converted back into the complex-valued ML solution as $[\hat\bms^\text{ML}]_i=[\bar{\bms}^\text{ML}]_i+j[\bar{\bms}^\text{ML}]_{i+U}$ for $i=1,\ldots,U$. 
For BPSK, the matrix $\bT=[\underline{\bH}^T\underline{\bH},-\underline{\bH}^T\overline{\bmy};-\overline{\bmy}^T\underline{\bH},\overline{\bmy}^T\overline{\bmy}]$ is of dimension $N\times N$ with $N=U+1$ and  $\tilde\bms=[\Re(\bms); 1]$. Here, we define  the $2B\times U$ matrix $\underline\bH=[  \Re(\bH) ;  \Im(\bH)]$.
Since $\Im(\bms)=\bZero$ in this case, $[\hat\bms^\text{ML}]_i=[\bar{\bms}^\text{ML}]_i$ for $i=1,\ldots,U$. 
In \secref{sec:sdrproblem}, we detail how the problem \rev{in} \eqref{eq:reformulatedMLproblem} can be relaxed into \rev{an SDP.}

\fussy

%%%
\subsection{Joint Channel Estimation and Data Detection}
\label{sec:SDRJED}

The second application scenario is \rev{JED} in large SIMO wireless uplink systems where one single-antenna user communicates over $K+1$ time slots with $\MR$ BS antennas. 
We use the following input-output relation to model the (narrow-band and flat-fading) SIMO wireless channel~\cite{stoica2003space,vikalo2006efficient,stoica2003joint,alshamary2015optimal}:
$\bY=\bmh\bms^H+\bN$.
Here, $\mathbf{Y}\in\mathbb{C}^{\MR\times (K+1)}$ contains the received vectors acquired over all $K+1$ time slots, $\bmh\in\mathbb{C}^{\MR}$ is the unknown SIMO channel vector that is assumed to be block fading, i.e., constant over $K+1$ time slots, $\bms^H\in\setO^{1\times (K+1)}$ is the transmit vector containing the data symbols from all $K+1$ time slots, and $\bN\in\mathbb{C}^{\MR \times (K+1)}$ is i.i.d.\ circularly-symmetric Gaussian with variance~$\No$ per entry. 
By assuming that~$\bmh$ is a deterministic but unknown channel vector with unknown prior statistics, we can formulate the following ML JED  problem~\cite{alshamary2015optimal}:
\begin{align} \label{eq:JEDproblem}
\{\hat{\bms}^\text{JED},\hat{\bmh}\} = \argmin_{\bms\in\setO^{K+1},\,\bmh\in\mathbb{C}^\MR} \|\bY-\bmh\bms^H\|_F.
\end{align}
It is important to note that there exists a phase ambiguity between both outputs of JED because~$\hat\bmh e^{j\phi}$ is also a solution whenever $\hat\bms^\text{JED}e^{j\phi}\in\setO^{K+1}$ for some phase $\phi$. As a consequence, one may convey information either as phase changes in the vector $\bms^H$ over time slots (known as differential encoding) or ``pin down'' the phase of one entry of the transmit vector; in what follows, we assume that the first transmitted entry is known to the receiver.\footnote{For SIMO systems, this approach resembles that of pilot-based transmission---the difference to JED is, however, that we also use all  transmitted information symbols to improve the channel estimate and hence, to improve the error-rate performance. } 

By assuming that the entries in $\bms$ are constant modulus (e.g., BPSK or QPSK), the ML JED estimate of the transmit vector reduces to \cite{alshamary2015optimal}:
\begin{align} \label{eq:reformulatedproblem}
\hat{\bms}^\text{JED} = \argmax_{\bms\in\setO^{K+1}}\, \|\bY\bms\|_2,
\end{align}
and $\hat{\bmh}=\bY\hat{\bms}^\text{JED}$ is the estimate of the channel vector. 
For a small number of time slots~$K+1$, the problem \rev{in}~\eqref{eq:reformulatedproblem} can be solved exactly at low average complexity using SD methods~\cite{alshamary2015optimal}. For systems with a large number of time slots, however, the computational complexity of such algorithms becomes prohibitive.
In contrast to \rev{the} coherent ML detection \rev{problem described in}~\eqref{eq:reformulatedMLproblem}, linear methods that approximate~\eqref{eq:reformulatedproblem} are unavailable  as relaxing the constraint $\bms\in\setO^{K+1}$ to $\bms\in\mathbb{C}^{K+1}$ causes the entries of $\bms$ to grow without bound.

\sloppy

We now show how the ML JED problem \rev{in} \eqref{eq:reformulatedproblem} can be transformed into the same structure of the coherent ML problem \rev{in}~\eqref{eq:reformulatedMLproblem}, which enables SDR. 
Since the receiver is assumed to know the first transmitted symbol $s_0$, we rewrite the objective in \eqref{eq:reformulatedproblem} as $\|\bY\bms\|_2=\|\bmy_0s_0+\bY_r\bms_r\|_2$ , where $\bY_r=[\bmy_1,\ldots,\bmy_K]$ and $\bms_r=[s_1,\ldots,s_K]^T$. 
Similarly to the coherent ML problem, we perform the real-valued decomposition by defining: 
\begin{align*}
\overline{\bmy} = \!\left[\begin{array}{c}
 \!\! \Re(\bmy_0s_0)  \!\!  \\
 \!\! \Im(\bmy_0s_0)  \!\!  
\end{array}\right]\!,  \,\,
\overline{\bH} = \!\left[\begin{array}{cc}
 \!\! \Re(\bY_r)  \!\!\!\!  &~~-\Im(\bY_r) \!\!  \\
 \!\! \Im(\bY_r)  \!\!\!\!  &~~\Re(\bY_r)  \!\! 
\end{array}\right]\!,  \,\,  % \text{and} \quad 
\overline{\bms} = \!\left[\begin{array}{c}
 \!\! \Re(\bms_r) \!\!  \\
 \!\! \Im(\bms_r)  \!\!  
\end{array}\right]\!,
\end{align*}
which allows us to rewrite $\|\bmy_0s_0+\bY_r\bms_r\|_2=\|\overline{\bmy}+\overline{\bH} \overline{\bms} \|_2$. 
We can now reformulate~\eqref{eq:reformulatedproblem} in a form that is equivalent to \eqref{eq:reformulatedMLproblem} as 
\begin{align} \label{eq:reformulatedJEDproblem}
\bar{\bms}^\text{JED}  =\argmin_{\tilde\bms\in\setX^{N}} \, \mathrm{Tr}(\tilde\bms^T\bT\tilde\bms).
\end{align}
For QPSK, the matrix $\bT=-[\overline{\bH}^T\overline{\bH},\overline{\bH}^T\overline{\bmy};\overline{\bmy}^T\overline{\bH},\overline{\bmy}^T\overline{\bmy}]$ is of dimension $N\times N$ with $N=2K+1$  and $\setX\in\{-1,+1\}$ with $\tilde\bms=[\Re(\bms_r); \Im(\bms_r); 1]$; for BPSK, the matrix $\bT=-[\underline{\bH}^T\underline{\bH},\underline{\bH}^T\overline{\bmy};\overline{\bmy}^T\underline{\bH},\overline{\bmy}^T\overline{\bmy}]$ is of dimension $N\times N$ with $N=K+1$ and  $\tilde\bms=[\Re(\bms_r); 1]$. Here, we define the $2B\times K$ matrix as $\underline\bH=[  \Re(\bY_r) ;  \Im(\bY_r)]$. 
Analogously to the coherent ML case, the solution $\bar{\bms}^\text{JED} $ can then be used to construct the complex-valued ML JED solution of \eqref{eq:reformulatedproblem}.

\fussy

\rev{Evidently, the problems described in \eqref{eq:reformulatedMLproblem} and \eqref{eq:reformulatedJEDproblem} exhibit the same structure---we next show how both of these problems can be solved approximately using the same SDR-based method.}

%%%
\subsection{Semidefinite Relaxation of the Problems \rev{in} \eqref{eq:reformulatedMLproblem} and \eqref{eq:reformulatedJEDproblem}}
\label{sec:sdrproblem}

SDR is a well-known approximation to the coherent ML problem~\cite{steingrimsson2003soft,luo2010semidefinite,tan2001application,ma2002quasi} and enables significantly lower (i.e., polynomial) computational complexity for systems employing BPSK and QPSK constellations.\footnote{SDR methods for higher-order constellations (such as 16-QAM) exist; see, e.g.,~\cite{wiesel2005semidefinite, sidiropoulos2006} for more details.} SDR not only provides near-ML performance, but also achieves the same diversity order as the ML detector~\cite{jalden2008diversity}. 
In contrast, the use of SDR for solving the ML JED problem as proposed in \secref{sec:SDRJED}  appears to be novel.

SDR-based data detection starts by reformulating the problems \rev{in} \eqref{eq:reformulatedMLproblem} and \eqref{eq:reformulatedJEDproblem} in the following equivalent form~\cite{luo2010semidefinite}:
\begin{align} \label{eq:SDRproblem}
\widehat{\bS} =\argmin_{\bS\in\mathbb{R}^{N\times N}}\, \mathrm{Tr}(\bT\bS) \,\, 
\text{ subject to } \mathrm{diag}(\bS) = \bOne, \, \mathrm{rank}(\bS) = 1.
\end{align}
Here, we used $\mathrm{Tr}(\bms^T\bT\bms)=\mathrm{Tr}(\bT\bms\bms^T)=\mathrm{Tr}(\bT\bS),$ where $\bS=\bms\bms^T$ is a rank-1 matrix and $\bms\in\setX^N$ is of appropriate dimension $N$. 
\rev{Unfortunately, the rank-one constraint in  \eqref{eq:SDRproblem} makes this problem at least as hard as the original two problems  in~\eqref{eq:reformulatedMLproblem} and \eqref{eq:reformulatedJEDproblem}.}
The key idea of SDR is to relax this rank constraint, which results in \rev{an} SDP that can be solved in polynomial time. 
Specifically, SDR applied to  \eqref{eq:SDRproblem} results in the following well-known optimization problem~\cite{luo2010semidefinite}: 
\begin{align} \label{eq:relaxedSDRproblem}
\widehat{\bS} =\argmin_{\bS\in\mathbb{R}^{N\times N}}\, \mathrm{Tr}(\bT\bS) \,\, 
\text{ subject to } \mathrm{diag}(\bS) = \bOne, \, \bS \succeq 0,
\end{align}
where the constraint $\bS \succeq 0$ ensures that the matrix~$\bS$ is positive semidefinite (PSD).  
If the result of the problem \rev{in} \eqref{eq:relaxedSDRproblem} is rank one, then $\widehat\bS=\hat\bms\hat\bms^H$ where $\hat\bms$ contains the exact estimate to  \eqref{eq:reformulatedMLproblem} and \eqref{eq:reformulatedJEDproblem}, i.e., SDR solves the original problem optimally. If the resulting matrix $\widehat\bS$ has a higher rank, then an estimate of the ML solution can be obtained by taking the signs of the leading eigenvector of~$\widehat\bS$ or by using  randomization schemes~\cite{luo2010semidefinite}. 

While~\eqref{eq:relaxedSDRproblem} can be solved exactly using interior-point methods~\cite{luo2010semidefinite}, \rev{such algorithms typically require (i) a large number of iterations, where each iteration requires either the solution to a linear system or an eigenvalue decomposition,} and (ii) high numerical precision, which renders fixed-point hardware challenging.  
We believe that these are the main reasons why---until now---\emph{no} VLSI  design of an SDR-based data detector has been proposed in the open literature. 

%%%%%%%%%%%%%%%%%%%%%%%%%%%%%%%%%%%%%%%%%%%%%%%%%%%%

\section{TASER: \mbox{Triangular Approximate Semidefinite Relaxation}}
\label{sec:algo}

We now detail TASER, a novel algorithm for approximately solving the SDP \rev{presented in} \eqref{eq:relaxedSDRproblem} \rev{using hardware accelerators.}

\subsection{Triangular SDP Formulation}
The key idea of TASER builds on the fact that real-valued PSD matrices $\bS \succeq 0$ can be factorized using the Cholesky decomposition $\bS=\bL^T\bL$, where~$\bL$ is an $N\times N$ lower-triangular matrix with non-negative entries on the main diagonal. With this result, we can reformulate the SDP \rev{shown in} \eqref{eq:relaxedSDRproblem} \rev{using} the following equivalent form:
\begin{align} \label{eq:triangularproblem}
\widehat{\bL} = \argmin_\bL\, \mathrm{Tr}(\bL\bT\bL^T) \,\,
\text{subject to } \|\bml_k\|_2 =1, \forall k.
\end{align}
Here, we replaced the constraint $\mathrm{diag}(\bL^T\bL) = \bOne$ of \eqref{eq:relaxedSDRproblem} by the equivalent $\ell_2$-norm constraint on the $k$th column $\bml_k=[\bL]_k$. 
To obtain (approximate) solutions to the  ML or JED ML problems in either \eqref{eq:reformulatedMLproblem} or \eqref{eq:reformulatedJEDproblem}, respectively, we can take the signs of the last row of the  solution matrix~$\widehat{\bL}$ from~\eqref{eq:triangularproblem}. 
\rev{In fact, if the solution matrix $\widehat\bS=\widehat{\bL}^T\widehat{\bL}$ has rank one (this implies that TASER identified the ML solution), then the last row of $\widehat{\bL}$ must contain the associated eigenvector as this is the only vector of dimension $N$. If, however, the solution matrix $\widehat\bS=\widehat{\bL}^T\widehat{\bL}$ has a higher rank, an approximate ML solution must be extracted somehow. As suggested in \cite{bach2005predictive,harbrecht2012low}, taking the last row of the Cholesky decomposition results in accurate rank-one approximations of PSD matrices. Our own simulations in \secref{sec:simresults} confirm that this approximation yields excellent error-rate performance, i.e., close to that of the exact SDR detector followed by an eigenvalue decomposition.
We emphasize that this approach avoids costly eigenvalue decompositions and randomization strategies that are required by conventional solvers that  compute~$\widehat\bS$ exactly using SDR.}

\subsection{Forward-Backward Splitting}
\label{sec:FBS}
We now develop a computationally efficient algorithm that directly solves the triangular SDP formulation in \eqref{eq:triangularproblem}. 
Unfortunately, the problem \rev{described in} \eqref{eq:triangularproblem} is non-convex in the matrix~$\bL$ and hence, computing an optimal solution is difficult. For TASER, we apply FBS~\cite{GSB14}, a computationally efficient method to solve convex optimization problems, to the non-convex problem \rev{in}~\eqref{eq:triangularproblem}. While this approach is not guaranteed to converge to the optimal solution of the non-convex problem \rev{posed by} \eqref{eq:triangularproblem}, we show in  \secref{sec:convergence} that TASER converges to a critical point of~\eqref{eq:triangularproblem}. Furthermore, our simulation results in \secref{sec:results} demonstrate near-ML error-rate performance.

FBS is an efficient, iterative method to solve convex optimization problems of the form $\hat{\bmx}=\argmin_{\bmx} f(\bmx) + g(\bmx)$, where the function~$f$ is smooth and convex,  and $g$ is convex but not necessarily smooth or bounded. FBS performs the following steps for $t=1,2,\ldots$~\cite{BT09,GSB14}: 
\begin{align*}
\bmx^{(t)} = \mathrm{prox}_g(\bmx^{(t-1)}-\tau^{(t-1)}\nabla f(\bmx^{(t-1)}) ;\tau^{(t-1)})
\end{align*}
until convergence or a maximum number of iterations $t_\text{max}$ is reached.
Here, $\{\tau^{(t)}>0\}$ is a suitably-chosen sequence of step size parameters, $\nabla f(\bmx)$ is the gradient of the function~$f$, and the so-called proximal operator for the function $g$ is defined as~\cite{BT09,GSB14}: 
  \begin{align} \label{eq:proximal}
\mathrm{prox}_g(\bmz;\tau) = \argmin_{\bmx} \,\bigl\{\tau g(\bmx) + \textstyle \frac{1}{2}\|\bmx-\bmz\|_2^2\bigr\}.
\end{align}

In order to approximately solve \eqref{eq:triangularproblem} using FBS, we define $f(\bL)=\mathrm{Tr}(\bL\bT\bL^T)$ and incorporate the constraint using $g(\bL)=\chi(\|\bml_k\|_2 =1, \forall k)$, where $\chi$ is the characteristic function (which is zero if the constraint is met and infinity otherwise). The gradient is given by $\nabla f(\bL)=\mathrm{tril}(2\bL\bT)$, where $\mathrm{tril}(\cdot)$ extracts the lower-triangular part of the argument. Even though the function $g$ is non-convex, the proximal operator \rev{defined in} \eqref{eq:proximal} has a closed form solution and is given by $\mathrm{prox}_g(\bml_k;\tau)=\bml_k/\|\bml_k\|_2$, $\forall k$; in words, the proximal operator simply rescales the columns of $\bL$ to have unit $\ell_2$-norm.

In order to arrive at a hardware-friendly algorithm, we avoid sophisticated \rev{step size} rules such as the ones proposed in~\cite{GSB14}. \rev{We use a fixed step size proportional to the reciprocal of the Lipschitz constant of the gradient $\nabla f(\bL)$ as proposed in~\cite{BT09}.} Our step size corresponds to $\tau=\alpha/\|\bT\|_2$, where~$\|\bT\|_2$ is the spectral norm of the matrix~$\bT$ and $0<\alpha<1$ is a system-dependent tuning parameter that we use to improve the empirical convergence rate when running TASER for a small number of iterations (see \secref{sec:convergence} for a discussion).  

\subsection{Jacobi Preconditioning} \label{sec:precon}
To improve the convergence rate of FBS, we precondition the problem \rev{presented in}~\eqref{eq:triangularproblem}. To this end, we compute a diagonal scaling matrix $\bD=\mathrm{diag}(\sqrt{T_{1,1}},\ldots,\sqrt{T_{M,M}})$, which we use to scale the matrix $\bT$ as $\widetilde{\bT}=\bD^{-1}\bT\bD^{-1}$ so that $\widetilde{\bT}$ has an all-ones main diagonal. The purpose of this so-called Jacobi preconditioner is to improve the condition number of the original PSD matrix~$\bT$~\cite{benzi2002preconditioning}. 
We then run FBS to recover a normalized version\footnote{In the conference paper~\cite{CGS2016a}, we mistakenly stated  $\widetilde{\bL}=\bD\bL$.} of the lower-triangular matrix $\widetilde{\bL}=\bL\bD$. We emphasize that preconditioning also requires us to modify the proximal operator, which turns out to be $\text{prox}_{\tilde{g}}(\tilde{\bml}_k)=D_{k,k}\tilde{\bml}_k/\|\tilde{\bml}_k\|_2 $, where $\tilde{\boldsymbol \ell}_k$ is the $k$th column of $\widetilde{\bL}$. Since we only rely on the signs of the last row of $\widehat{\bL}$ to obtain an estimate of the ML problems, we can simply take the signs of the normalized triangular matrix $\widetilde{\bL}$.

%%%
\subsection{The TASER Algorithm}
\begin{algorithm}[t]
\caption{TASER \label{alg:TASER}}
\begin{algorithmic}[1]
\STATE \textbf{inputs:} $\widetilde{\bT}$,  $\bD$, and $\tau=\alpha/\|\widetilde{\bT}\|_2$
\STATE \textbf{initialization:} $\widetilde{\bL}^{(0)} = \mathbf{D}$
\FOR{$t = 1,\ldots, t_\text{max}$}
\STATE $\bV^{(t)} = \widetilde{\bL}^{(t-1)} - \mathrm{tril}(2\tau\widetilde{\bL}^{(t-1)}\widetilde{\bT})$
\STATE $\widetilde{\bL}^{(t)} = \mathrm{prox}_{\tilde{g}}(\bV^{(t)})$
\ENDFOR
\STATE \textbf{outputs:} $\overline{s}_k = \mathrm{sign}(\widetilde{L}^{(t_\text{max})}_{N,k}), k=1,\ldots, N-1$
\end{algorithmic}
\end{algorithm}
We now have all the necessary ingredients for TASER, which is  summarized in~\algref{alg:TASER}. The inputs of the algorithm are the preconditioned matrix $\widetilde{\bT}$, the scaling matrix $\bD$, and the step size $\tau$. We  initialize the FBS procedure by $\widetilde{\bL}^{(0)} = \mathbf{D}$, which resulted in excellent performance for all considered scenarios. The main loop of TASER then performs the gradient and proximal steps as discussed in Sections~\ref{sec:FBS} and~\ref{sec:precon} until a maximum number of iterations $t_\text{max}$ is reached. For most situations, only a few iterations are sufficient to achieve near-ML error rate performance (see \secref{sec:results} for numerical results). The TASER algorithm computes an estimate for the coherent ML and ML JED problems \rev{in} \eqref{eq:reformulatedMLproblem} and \eqref{eq:reformulatedJEDproblem}, respectively.

%%%
\subsection{Convergence Theory}
\label{sec:convergence}

The TASER algorithm tries to solve a non-convex problem using FBS.  Hence, our approach raises two questions, namely (i) whether we should expect the minimization algorithm to converge, and (ii) whether the local minima of the non-convex problem correspond to minimizers of the convex SDP. We now investigate both of these questions.

While the application of FBS for minimizing the proposed semidefinite program is new, the convergence of FBS for non-convex problems is well-studied.  Reference~\cite{attouch2013convergence} presents conditions for which FBS converges with non-convex constraints.  In particular, the problem must be semi-algebraic, meaning both the constraints and the epigraph of the objective can be written as the set of solutions to a system of polynomial equations.\footnote{The authors of \cite{attouch2013convergence} actually prove results for the broader class of Kurdyka-\L{}ojasiewicz functions, of which semi-algebraic functions are a special case.}   Fortunately, such results apply to the formulation \eqref{eq:triangularproblem}. The following result makes this statement rigorous.

\begin{prop}\label{prop:converge}
Suppose we apply FBS (Algorithm \ref{alg:TASER}) to solve the problem \rev{stated in} \eqref{eq:triangularproblem}. If we use the \rev{step size} $\tau=\alpha/\|\bT\|_2$ with $0<\alpha<1,$ then the sequence of iterates $\{\bL^{(t)}\}$ converges to a critical point of the problem \rev{in} \eqref{eq:triangularproblem}.  
\end{prop}
\begin{proof}
The function $\|\boldsymbol\ell_k\|^2_2$  is a polynomial in the entries of~$\bL_{}$. The constraint set in  \eqref{eq:triangularproblem}  is the solution to the polynomial system $\|\boldsymbol\ell_k\|^2_2=1$, $\forall k$ and is thus semi-algebraic.  The objective function, being a quadratic form, is also trivially semi-algebraic.  By Theorem~5.3 of \cite{attouch2013convergence}, we know that the sequence of iterates $\{{\bL}^{(t)}\}$ converges, provided the \rev{step size} is bounded from above by the inverse of the Lipschitz constant of the gradient of the objective.  For our quadratic objective, the Lipschitz constant is merely the spectral radius ($\ell_2$-norm) of ${\bT}.$
\end{proof}

Note that the Jacobi preconditioner in \secref{sec:precon} results in a problem of the same form as~\eqref{eq:triangularproblem}, but with constraints of the form $\|\tilde{\boldsymbol\ell}_k\|^2_2 = D_{k,k}^2$ and the step size $\tau=\alpha/\|\widetilde{\bT}\|_2$. Consequently, Proposition~\ref{prop:converge} still applies. Note that this result has the caveat that we are not guaranteed to find a (global) minimizer, but rather stationary points, although we generally observe minimizers in practice.  Nonetheless, this convergence guarantee is considerably stronger than what is known for other low-complexity SDP methods, such as those inspired by Burer and Montiero \cite{burer2003nonlinear}, \rev{which rely on non-convex augmented Lagrangian schemes for which no guarantees currently exist.}

The second question to ask is whether the local minima of our non-convex formulation \rev{in}  \eqref{eq:triangularproblem} correspond to minimizers of the convex SDP \rev{shown in} \eqref{eq:relaxedSDRproblem}.  Interestingly, when the factors $\bL$ and $\bL^T$ are not constrained to be triangular, local minimizers of \eqref{eq:triangularproblem} are known to yield optimal minimizers for the SDP \eqref{eq:relaxedSDRproblem}  (see \cite{boumal2015riemannian}, Corollary~3.6).  Nevertheless, we have found that it is better to enforce the triangular constraint in practice as it substantially simplifies the architecture detailed next. 

%%%%%%%%%%%%%%%%%%%%%%%%%%%%%%%%%%%%%%%%%%%%%%%%%%%%

\section{VLSI Architecture}
\label{sec:impl}

We now propose a systolic VLSI architecture that implements TASER and enables high-throughput data detection at low hardware complexity.

\subsection{Architecture Overview}
\begin{figure}[tp]
\centering
\includegraphics[width=0.75\columnwidth]{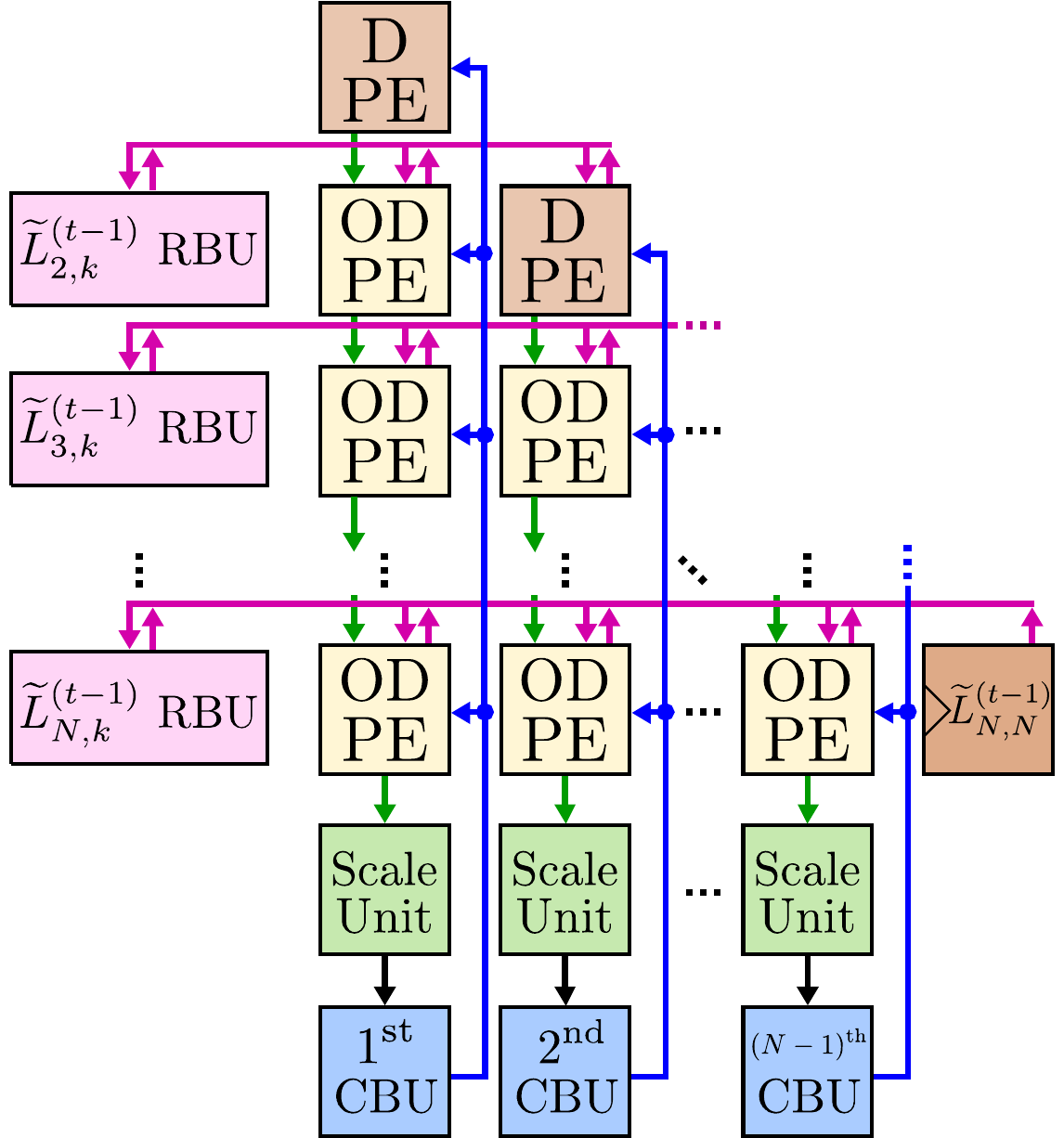}
\caption{High-level block diagram of TASER. We use a systolic array of processing elements (PEs) for the diagonal (D) and off-diagonal (OD) elements, which enables high throughput at low hardware complexity.}
\label{fig:array}
\end{figure}

\figref{fig:array} shows the proposed triangular systolic array consisting of ${N(N+1)}/{2}$  processing elements (PEs), which mainly perform multiply-accumulate (MAC) operations. Each PE is associated with an entry~$\widetilde{L}_{i,j}^{(t-1)}$ of the lower-triangular matrix~$\widetilde{\bL}^{(t-1)}$ and stores~$\widetilde{L}_{i,j}^{(t-1)}$ as well as the value $V_{i,j}^{(t)}$ of the $\bV^{(t)}$ matrix (cf.~\algref{alg:TASER}).
All PEs that are part of the same column receive data from a column-broadcast unit~(CBU); all  PEs that are part of the same row receive data from a row-broadcast unit (RBU). 

In the $k$th cycle during the $t$th TASER iteration, the $i$th RBU sends the value $\widetilde{L}_{i,k}^{(t-1)}$ to all PEs on row~$i$, while the $j$th CBU sends~$\widehat{T}_{k,j}$ to all PEs on column~$j$.
We assume that the (scaled) matrix~$\widehat{\bT}=2\tau\widetilde{\bT}$ has been computed in a pre-processing step and is stored in a memory (see \secref{sec:details} for more details on the memory implementation). The $\widetilde{L}_{i,k}^{(t-1)}$ value coming from each RBU is taken from the $(i,k)$ PE and sent to all other PEs in the same row.

With the data received from the CBU and the RBU, each PE performs MAC operations until the result $\widetilde{\bL}^{(t-1)} \widehat{\bT}$ on line 4 of \algref{alg:TASER} is computed. To include the subtraction on line~4, the operation $\widetilde{L}_{i,j}^{(t-1)} - \widetilde{L}_{i,1}^{(t-1)} \widehat{T}_{1,j}$ is performed in the first cycle of each TASER iteration and stored in the accumulator. During subsequent cycles, the products $\widetilde{L}_{i,k}^{(t-1)} \widehat{T}_{k,j}$, with $2 \leq k \leq N$, are sequentially subtracted from the accumulator.
Since the matrix $\widetilde{\bL}$ is lower-triangular, we have $\widetilde{L}_{i,k'} = 0$ if $i < k'$. Hence, we avoid the subtraction of $\widetilde{L}_{i,k'}^{(t-1)} \widehat{T}_{k',j}$ as they are zero. This implies that the $V^{(t)}_{i,j}$ values of the PEs in the $i$th row of the systolic array are computed after only~$i$ clock cycles, so the matrix $\bV^{(t)}$ on line~4 is completed after~$N$ cycles.

An example for an  $N=3$ array is shown in \figref{fig:archexpl_cyc1}. In the first cycle of the $t$th iteration, the $(1,1)$ PE has access to the values $\widetilde{L}_{1,1}^{(t-1)}$ and $\widehat{T}_{1,1}$, so it can compute $V_{1,1}^{(t)} = \widetilde{L}_{1,1}^{(t-1)} - \widetilde{L}_{1,1}^{(t-1)} \widehat{T}_{1,1}$. In the same cycle, the PEs on the second row perform their first MAC operation, which leaves $\widetilde{L}_{2,j}^{(t-1)} - \widetilde{L}_{2,1}^{(t-1)} \widehat{T}_{1,j}$ in their accumulators. 

In the second cycle, the PEs on the second row receive $\widetilde{L}_{2,2}^{(t-1)}$ via the RBU and $\widehat{T}_{2,j}$ via the CBUs (see \figref{fig:archexpl_cyc2}), so they can finish computing $V_{2,j}^{(t)} = \widetilde{L}_{2,j}^{(t-1)} - \widetilde{L}_{2,1}^{(t-1)} \widehat{T}_{1,j} - \widetilde{L}_{2,2}^{(t-1)} \widehat{T}_{2,j}$. 
In addition, in this same cycle, the (1,1) PE can use its MAC unit to square its $V_{1,1}^{(t)}$ value. This result will be available and sent to the next PE in the same column on the following cycle, which is represented with the green arrow in \figref{fig:archexpl_cyc3}.

In the third cycle, the $(2,1)$ PE  has access to ${V_{1,1}^{(t)}}^2$ (from the $(1,1)$ PE) and $V_{2,1}^{(t)}$ (stored internally), so it can use its MAC unit to square $V_{2,1}^{(t)}$ and add it to ${V_{1,1}^{(t)}}^2$ (see \figref{fig:archexpl_cyc3}). 
The result will be the sum of the squares of the first two elements of the first column of $\bV^{(t)}$ and, in the next cycle, the result will be available and sent to the next PE in the same column (for this example, the $(3,1)$ PE), so it can repeat the same procedure. This process is replicated in all the columns and repeated until all the PEs of the array have completed their calculations. By doing so, the squared $\ell_2$-norm of each column of $\bV^{(t)}$ is computed after~$N+1$ clock cycles, just one cycle after $\bV^{(t)}$ is completed.
In the $(N+2)$th cycle, the squared $\ell_2$-norm for the $j$th column is passed to a scale unit, which computes its inverse square root and multiplies the result with $D_{j,j}$. This operation takes two clock cycles to complete, so its result is ready in the $(N+4)$th cycle.

\sloppy

\rev{In the $(N+4)$th cycle, the scaling factor $D_{j,j}/\|\bmv_j\|_2$ (where~$\bmv_j$ is the $j$th column of $\bV^{(t)}$) is sent to all the PEs in the same column via the CBU, as shown in \figref{fig:archexpl_cyc4}.} 
Then, in the $(N+5)$th and final cycle of the iteration, all PEs multiply the received scaling factor to their associated~$V_{i,j}^{(t)}$ value to obtain the next iterate $\widetilde{L}_{i,j}^{(t)}$, thus completing the proximal step on line~5 of \algref{alg:TASER}.

\fussy

Prior to decoding the next symbol, line~2 of \algref{alg:TASER} must be executed; this is accomplished using the CBUs, which send the $D_{j,j}$ values to the diagonal PEs, while the off-diagonal PEs clear their $\widetilde{L}_{i,j}^{(t-1)}$ registers.

\begin{figure}[tp]
\begin{tabular}{cc}
\subfigure[First cycle]{\includegraphics[width=0.22 \columnwidth]{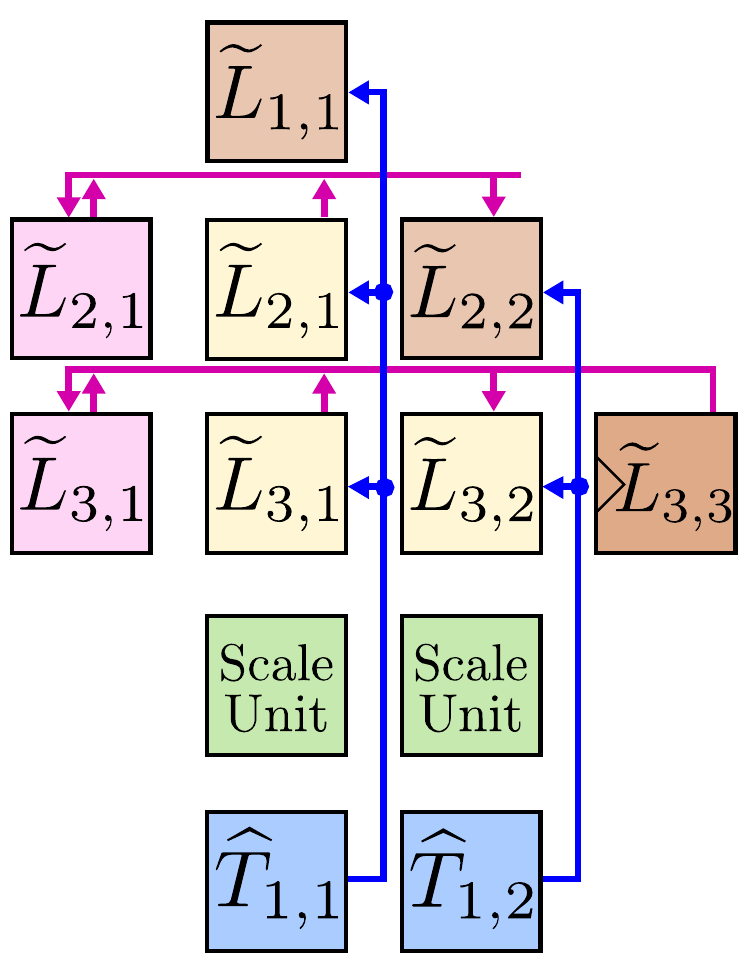} \label{fig:archexpl_cyc1}} 
\hspace{0.0cm}
\subfigure[Second cycle]{\includegraphics[width=0.22 \columnwidth]{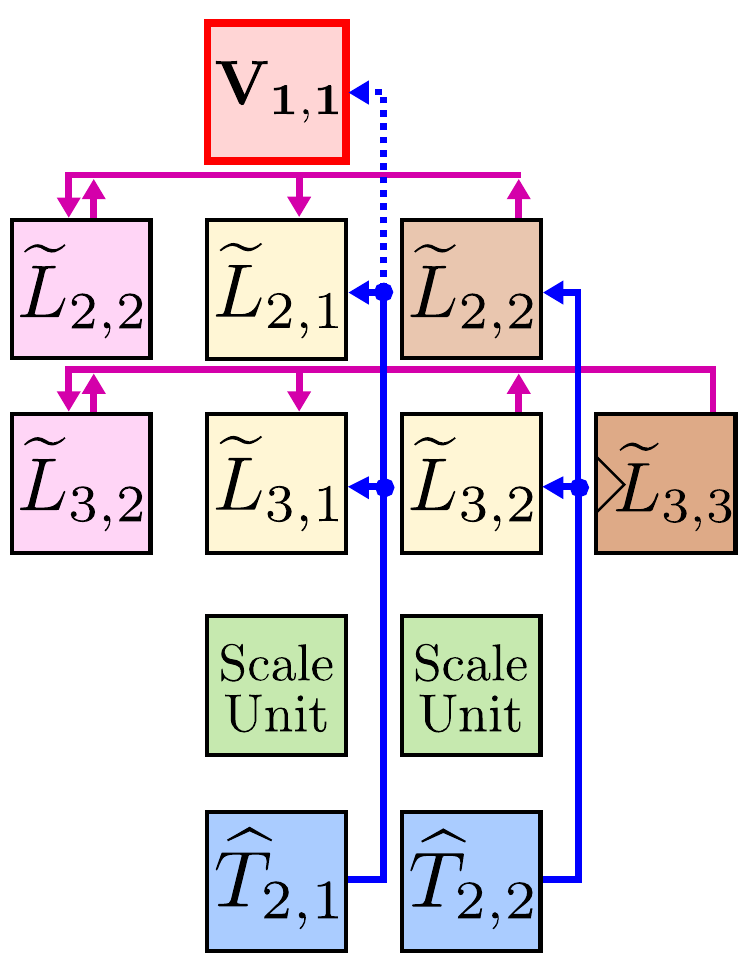}\label{fig:archexpl_cyc2}}
\hspace{0.0cm}
\subfigure[Third cycle]{\includegraphics[width=0.22 \columnwidth]{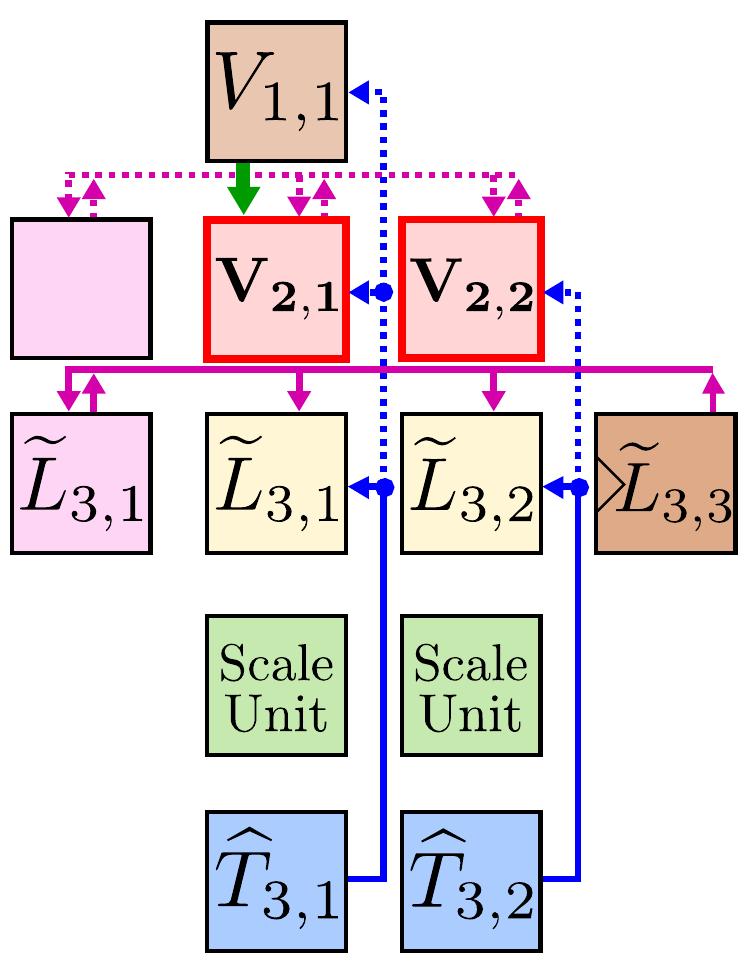}\label{fig:archexpl_cyc3}}		
\hspace{0.0cm}
\subfigure[Seventh cycle]{\includegraphics[width=0.22 \columnwidth]{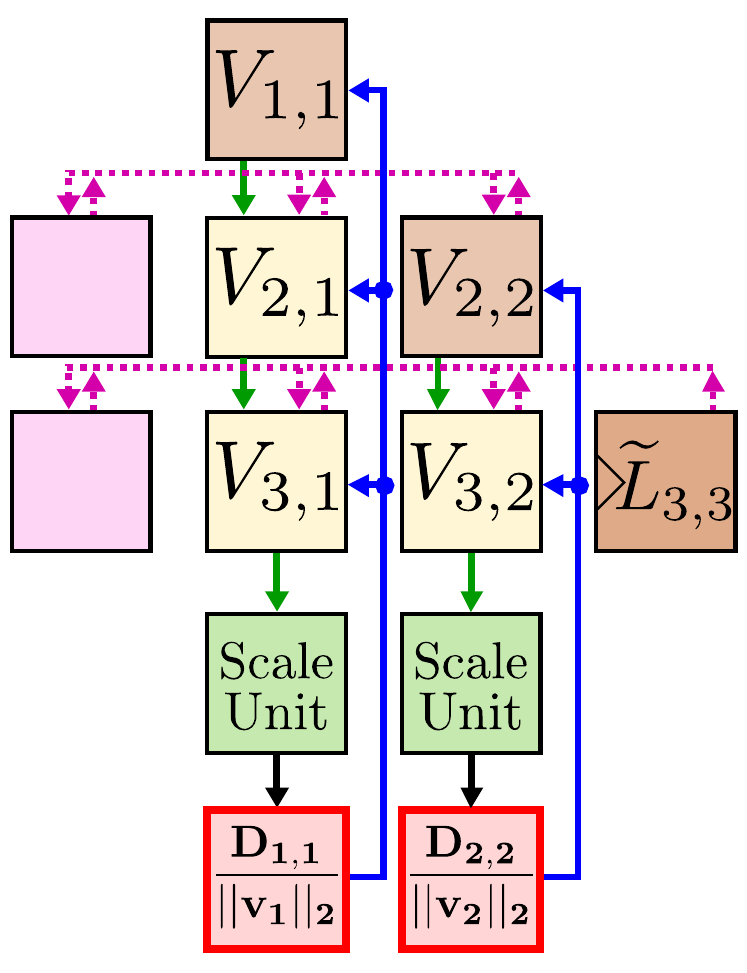}\label{fig:archexpl_cyc4}}		
\end{tabular}
\caption{Different cycles for the $t$th iteration on an $N=3$ TASER array. The symbols inside the PEs correspond to the quantities of interest, for each cycle, stored in each PE, while the symbols inside the RBUs and CBUs correspond to the quantities being transmitted by these units. All the $\widetilde{L}$ values correspond to the $(t-1)$th iteration, while the $V$ values are from the $t$th iteration.}
\end{figure}

\begin{figure}[tp]
\centering
\includegraphics[width=0.95\columnwidth]{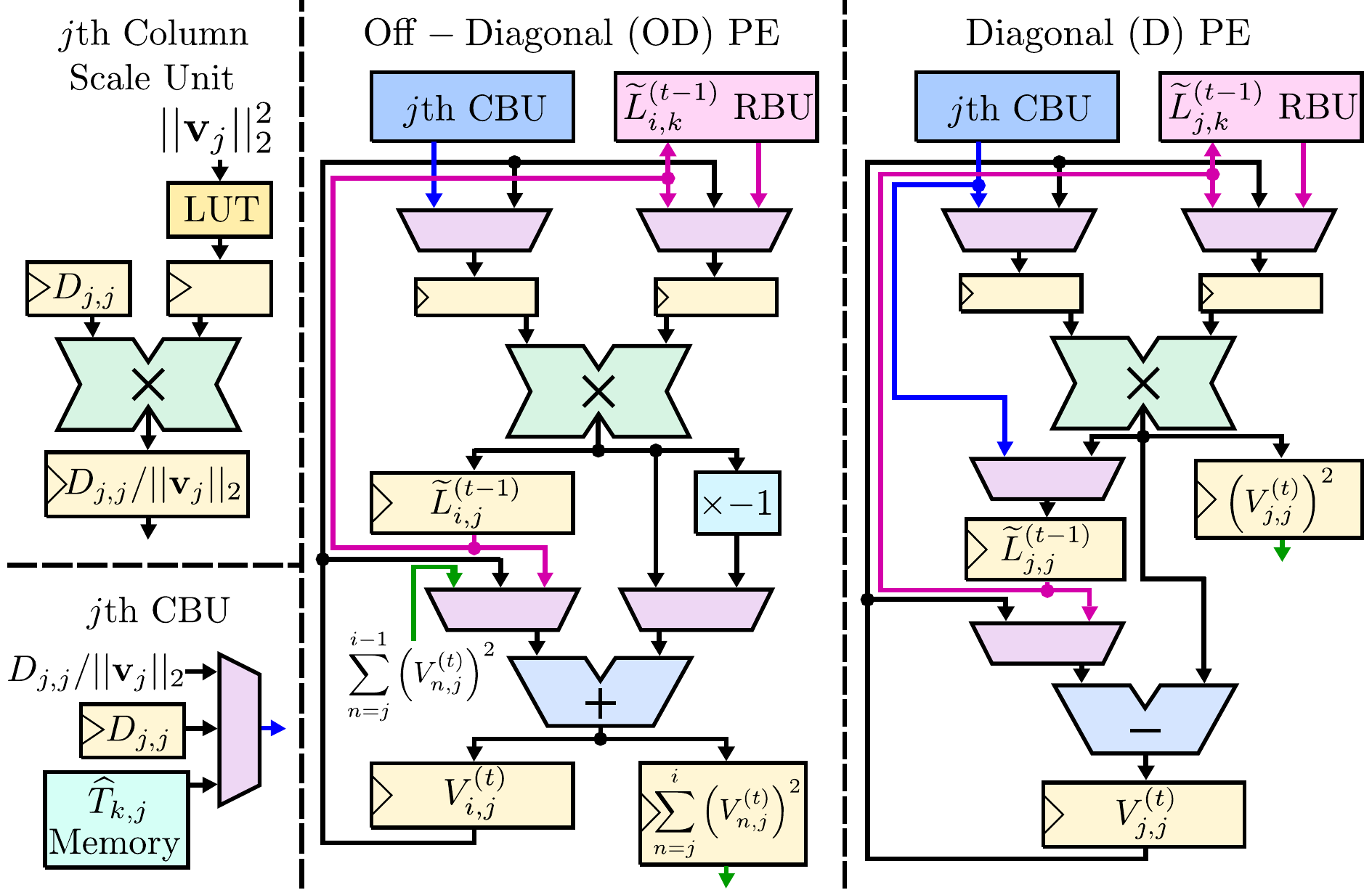}
\caption{Architecture details of the column-scale unit (CSU), the column-broadcast unit (CBU),  and the off-diagonal (OD) and diagonal (D) processing elements (PEs).}
\label{fig:pes}
\end{figure}

\subsection{Processing Element}

We use two slightly distinct types of PEs in our systolic array: (i) off-diagonal~(OD) PEs and (ii) diagonal~(D) PEs (see Figure \ref{fig:pes}). Both PE types support the following four operation modes:

\subsubsection{Initialization of $\widetilde{\bL}$}
This mode is used for line 2 of \algref{alg:TASER}. All off-diagonal PEs initialize $\widetilde{L}_{i,j}^{(t-1)}=0$; the diagonal PEs initialize their states with~$D_{j,j}$ received from the~CBU.

\subsubsection{Matrix multiplication}
This mode is used to compute line 4 of \algref{alg:TASER}.
The multiplier uses the inputs from both broadcast signals. In the first cycle of the matrix-matrix multiplication procedure, the multiplier's output is subtracted from $\widetilde{L}_{i,j}^{(t-1)}$; in all other cycles, it is subtracted from the accumulator. Since each PE stores its own $\widetilde{L}_{i,j}^{(t-1)}$ value, in the $k$th cycle, all the PEs in the $k$th column use their internal $\widetilde{L}_{i,k}^{(t-1)}$ value to feed the multiplier, instead of the signals coming from the RBU. 

\subsubsection{Squared $\ell_2$-norm calculation}
This mode is used for line 5 of \algref{alg:TASER}. Both of the multiplier's inputs are $V_{i,j}^{(t)}$. For the D-PEs, the result is passed to the next PE in the same column. For the OD-PEs, the output of the multiplier is added to the $\sum_{n=j}^{i-1} \left(V_{n,j}^{(t)}\right)^2$ value from the preceding PE in the same column; the result $\sum_{n=j}^{i} \left(V_{n,j}^{(t)}\right)^2$ is sent to the next PE or to the scale unit, if the PE is in the last row. 

\subsubsection{Scaling}
This mode \rev{completes} line 5 of \algref{alg:TASER}.
One of the multiplier's inputs is $V_{i,j}^{(t)}$ and the other is the value $D_{j,j}/\|\bmv_j\|_2$, which was computed previously by the scale unit and received through the CBU. The result is $\widetilde{L}_{i,j}^{(t)}$ and is stored in every PE as the $\widetilde{L}_{i,j}^{(t-1)}$ of the next iteration.

\subsection{Implementation Details}
\label{sec:details}
To demonstrate the efficacy of TASER and the proposed triangular systolic array, we implemented FPGA and ASIC reference designs for various array sizes $N$. \rev{All designs were developed and optimized in Verilog on register-transfer level (RTL).} The implementation details are as follows:

\subsubsection{Fixed-point design parameters}
To minimize the hardware complexity while maintaining near-optimal error-rate performance, all our designs use $14$\,bit fixed-point numbers. All PEs, except for the ones in the bottom row of the triangular array, use $8$ fraction bits to represent $\widetilde{L}_{i,j}^{(t-1)}$ and $V_{i,j}^{(t)}$; the PEs in the bottom row use $7$ fraction bits. For the element $\widetilde{L}_{N,N}$, we do not use a PE and store the value (which remains constant) in a register with~$5$ fraction bits.

\begin{figure*}[tp]
\centering
\subfigure[BPSK]{\includegraphics[height=5cm]{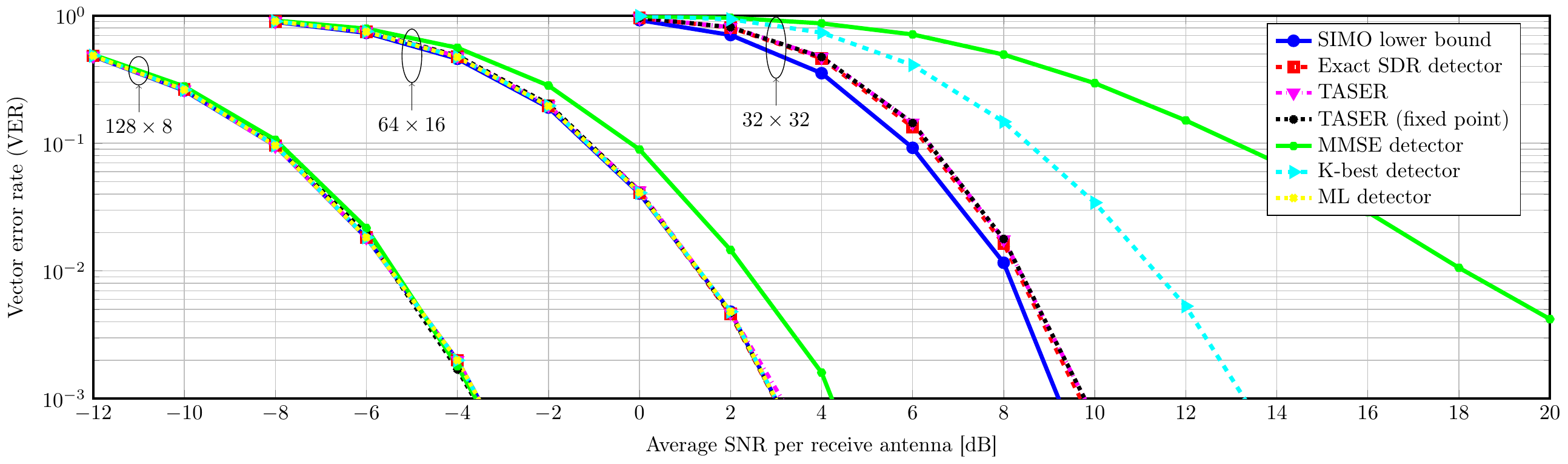}\label{fig:16x16bpsk_ver}}
\subfigure[QPSK]{\includegraphics[height=5cm]{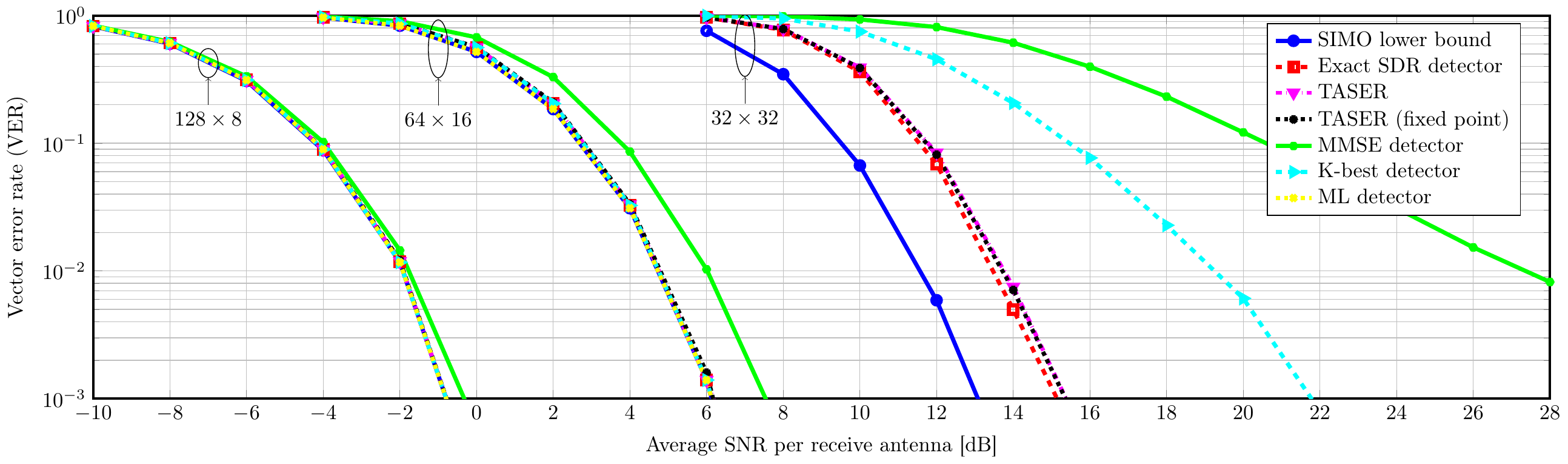}\label{fig:16x16qpsk_ver}}
\caption{\rev{Vector error rate (VER) for three different $B \times U$ large-MIMO system configurations. Even for square massive MU-MIMO systems (see the $32\times32$ system), TASER achieves near-optimal VER performance (close-to-ML and the SIMO lower bound) and approaches the  performance of the exact SDR detector. For systems with more BS antennas than users (see the $128\times8$ system), all of the considered data detectors approach optimal performance.}}
\label{fig:16x16_ver}
\end{figure*}

\begin{figure}[tp]
\centering
\subfigure[BPSK]{\includegraphics[height=5cm]{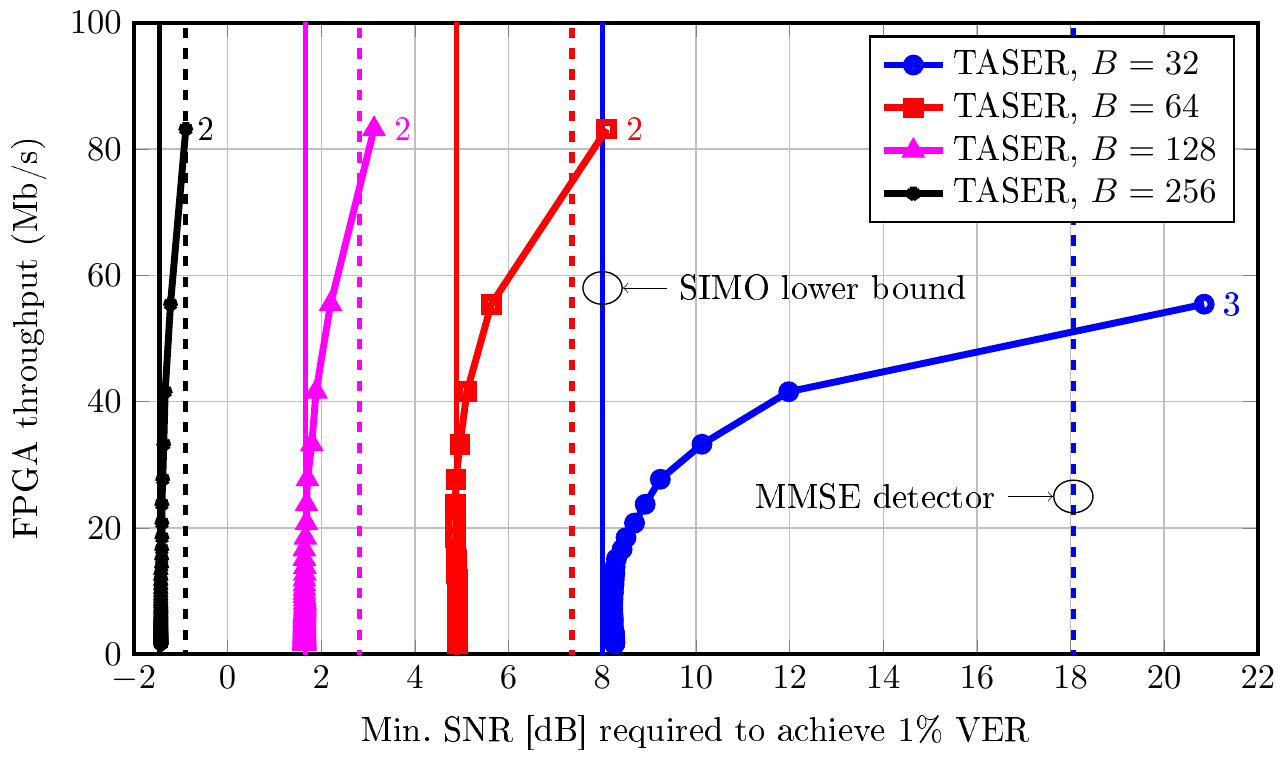}\label{fig:16x16bpsk_tradeoff}}
\subfigure[QPSK]{\includegraphics[height=5cm]{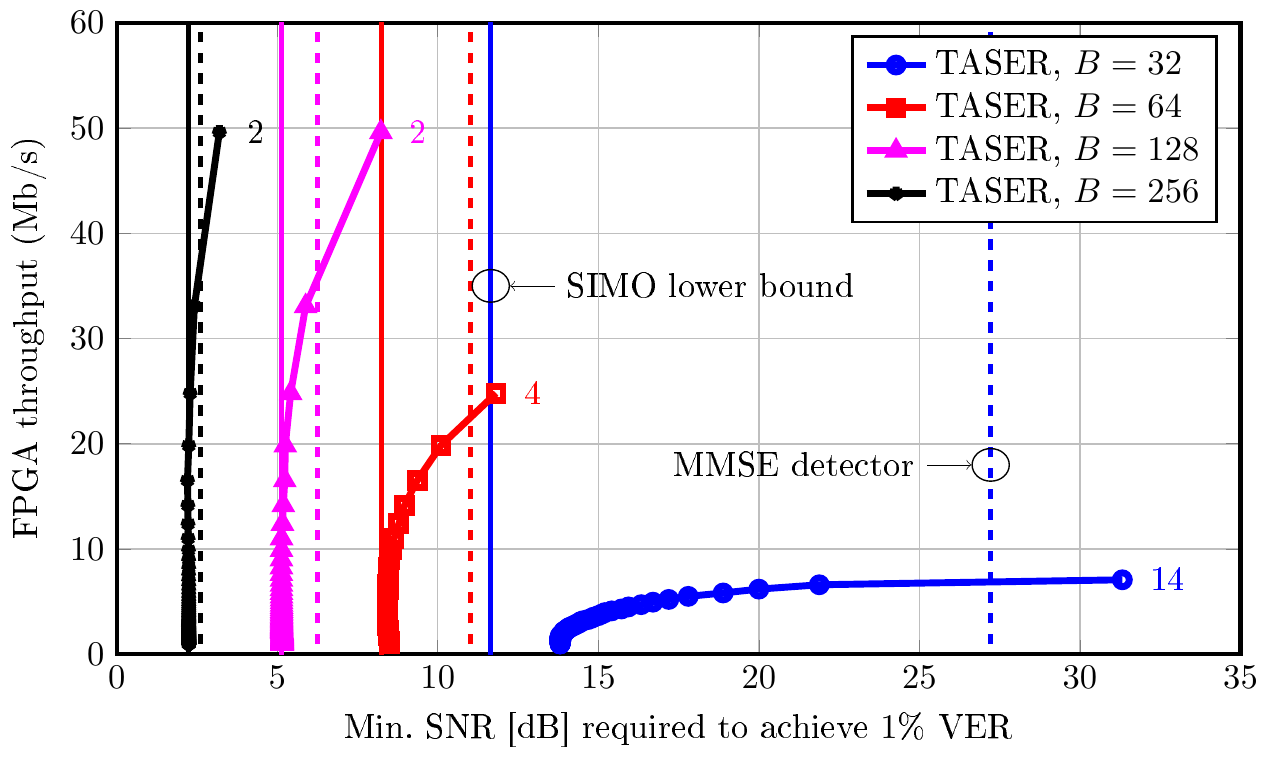}\label{fig:16x16qpsk_tradeoff}} \\[-0.075cm]
\caption{Throughput for the \rev{FPGA} design vs.\ performance trade-off for a \rev{$32$-user} system. Vertical \rev{solid} lines represent the SIMO lower bound; dashed lines represent linear MMSE performance. TASER outperforms linear detectors in almost all operation regimes. The numbers next to the markers correspond to the number of TASER iterations $t_\text{max}$.}
\vspace{-0.15cm}
\end{figure}

\begin{figure}[tp]
\centering
\subfigure[BPSK]{\includegraphics[height=5cm]{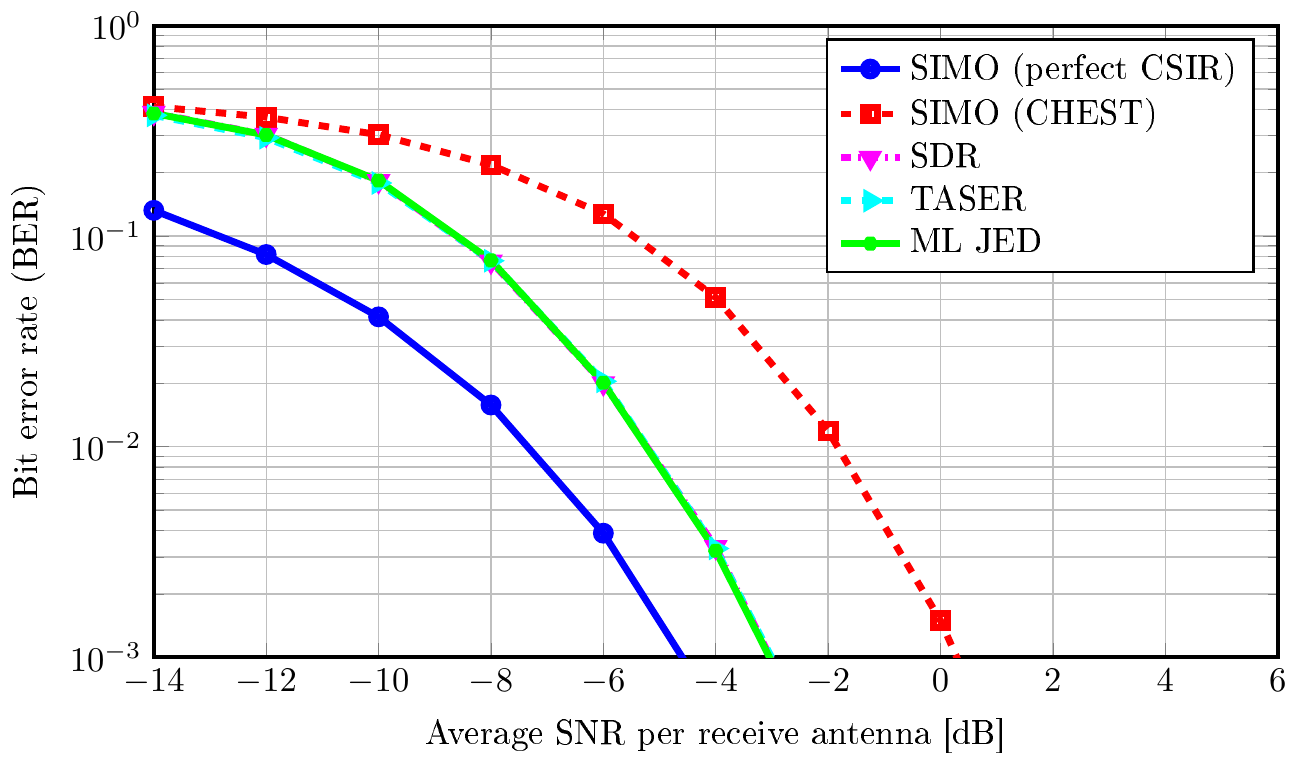}\label{fig:16_jed_bpsk_ver}}
\subfigure[QPSK]{\includegraphics[height=5cm]{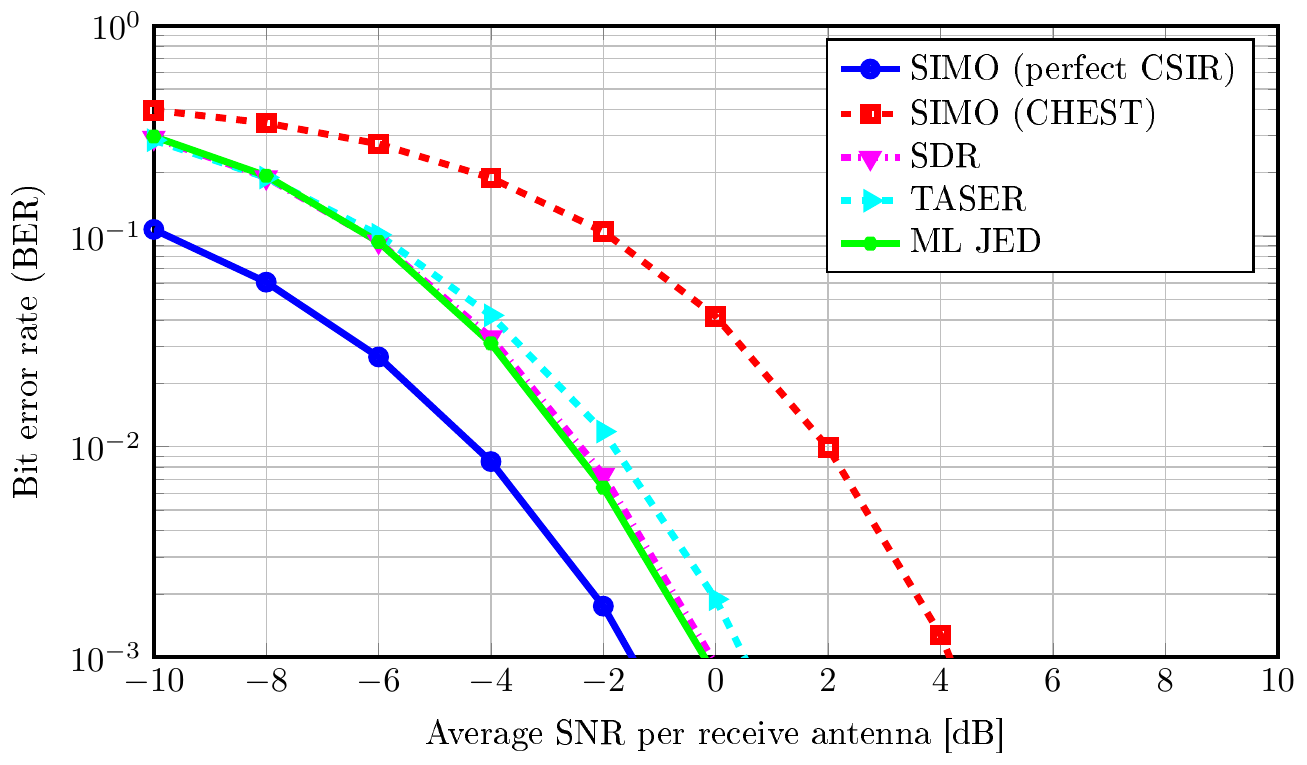}\label{fig:16_jed_qpsk_ver}}
\caption{\rev{Bit error rate (BER)} for a SIMO system with $16$ BS antennas with transmission over $16$ time slots. TASER-based  JED achieves near-optimal BER performance (close-to perfect CSIR) and achieves performance similar to the exact SDR detector and ML JED; channel estimation (CHEST) followed by SIMO detection entails a 3\,dB SNR loss.}
\end{figure}

\subsubsection{Inverse square-root computation}
The inverse square-root operation in the scale unit is implemented using a look-up table (LUT), which we synthesized using random logic. Each LUT consists of $2^{11}$ entries with $14$\,bits per word, of which $13$ are fraction bits.

\subsubsection{\rev{$\widehat{\bT}$}-matrix memories}
For the FPGA designs, the $\widehat{T}_{k,j}$ memories are implemented with LUTs used as distributed RAM (i.e., no block RAMs were used); for the ASIC designs, we use latch arrays built from standard cells~\cite{meinerzhagen2010towards} in order to minimize the circuit area.

\subsubsection{RBU and CBU design}
\rev{The RBUs are implemented differently for the FPGA and ASIC designs. For the FPGA designs, the RBU of the $i$th row is an $i$-input multiplexer that receives data from all the PEs on its row, and also sends the appropriate $\widetilde{L}_{i,k}^{(t-1)}$ to these PEs. For the ASIC designs, the RBU consists of a bidirectional bus, where each PE on its row uses a tri-state buffer to send data through it one at a time, while all the PEs on the same row acquire data from it. A similar approach is used for the CBUs:} We use multiplexers for the FPGA designs and busses for the ASIC designs.
For both target architectures, the output of the $i$th RBU connects to $i$ PEs. This path suffers from large fan-out for large values of $i$, eventually becoming the critical path for large systolic arrays. The same behavior applies to the CBUs.
In order to shorten these critical paths in our architecture, we place stage registers at the inputs and outputs of the respective broadcast units. While this approach entails two penalty cycles per TASER iteration, the overall detection throughput is increased as we achieve a substantially higher clock frequency.

\section{Implementation Results and Comparison}
\label{sec:results}

We now provide error-rate performance results for coherent data detection in massive MU-MIMO systems and for JED in massive SIMO systems. We then show reference FPGA and ASIC implementation results which we compare to existing designs for massive MU-MIMO systems.

\subsection{Error-Rate Performance}
\label{sec:simresults}

\subsubsection{Coherent massive MU-MIMO data detection}
\rev{Figures~\ref{fig:16x16bpsk_ver} and \ref{fig:16x16qpsk_ver} show vector error rate (VER) simulation results for TASER  with BPSK and QPSK modulation, respectively.\footnote{\rev{The vector error rate (VER) corresponds to $\Pr{\bms\neq\hat{\bms}}$, which is the probability of detecting a different vector $\hat{\bms}$ than the transmitted one $\bms$.}} We show simulation results for coherent data detection  with i.i.d.\ flat Rayleigh fading in tall $128\times8$ and $64\times16$ systems, as well as a square $32\times32$ large MU-MIMO system (we use the notation $\MR\times \MT$).
We show the performance of ML detection (only for the $U=8$ and $U=16$ systems; computed using the sphere-decoding algorithm in \cite{burg2005vlsi}), exact SDR detection from~\eqref{eq:SDRproblem}, linear MMSE detection, and the $K$-best algorithm as detailed in~\cite{wenk2006k} with $K=5$. 
As a baseline, we also include the performance of the SIMO lower bound.}

\rev{For the $128\times8$ massive MIMO system, we see that all detectors approach optimal performance (even the SIMO lower bound); this is a well-known result from the large MIMO literature \cite{Marzetta2010,Rusek2012,hoydis2013massive,larsson2014massive}.
For the $64\times16$ massive MIMO system, only the linear MMSE detector suffers from a (rather small) performance loss; all the other detectors perform equally well. 
For the more challenging square $32\times32$ massive MIMO system, we see that TASER achieves near-ML performance and outperforms linear MMSE detection and the $K$-best algorithm (note that, even with the sphere decoder, ML detection exhibits prohibitive complexity). 
We also show the fixed-point performance of our TASER hardware design, denoted by ``fixed point'' in Figures~\ref{fig:16x16bpsk_ver} and \ref{fig:16x16qpsk_ver}, which demonstrates a small implementation loss (less than 0.2\,dB SNR at 1\% VER).}

Figures \ref{fig:16x16bpsk_tradeoff} and \ref{fig:16x16qpsk_tradeoff} show the trade-off between the throughput of TASER for the \rev{FPGA} design \rev{(see \secref{sec:FPGAdesign} for the details)} and the minimum SNR required to achieve $1$\% VER for the coherent data detection in large-MIMO systems. 
We also include the SIMO lower bound and the performance of linear MMSE detection as a reference. \rev{The MMSE detector} serves as a fundamental performance limit of the conjugate gradient least-squares (CGLS) detector~\cite{Yin2015}, the Neumann-series detector~\cite{Wu2014},  the \rev{optimized coordinate-descent (OCD)} detector~\cite{Wu2016}, and the Gauss-Seidel (GS) detector~\cite{WZXXY2016}.
The maximum number of TASER iterations~$t_\text{max}$ enables us to tune the performance/complexity trade-off; only a few iterations are sufficient to outperform linear detection. 
\rev{We also see that TASER delivers near-ML performance and  achieves throughputs from $10$\,Mb/s to $80$\,Mb/s for the FPGA design.}

\subsubsection{JED in massive SIMO systems}
Figures \ref{fig:16_jed_bpsk_ver} and \ref{fig:16_jed_qpsk_ver} show \rev{BER} simulation results for TASER with BPSK and QPSK modulation, respectively. The simulations are for a $B=16$ BS antennas and $16$ time slots SIMO system; we perform $t_\text{max}=20$ iterations and use an i.i.d.\ flat Rayleigh block-fading channel model.
We include the performance of the SIMO detection with both perfect receiver channel state information (CSIR) and channel estimation (CHEST), exact SDR detection \rev{from} \eqref{eq:SDRproblem}, and ML JED detection (which is computed using the algorithm proposed in~\cite{xu2008exact}). 
We see that TASER achieves near-optimal performance, as it is close to a system with perfect CSIR, and outperforms detection via SIMO CHEST, while achieving similar performance as ML JED and exact SDR detection at manageable complexity.
We note that the trade-offs between throughput and SNR performance behave analogously to the massive MU-MIMO case. 

\subsection{\rev{Computational Complexity}}
\rev{We now compare the computational complexity of TASER with other large-scale MIMO data-detection algorithms proposed in the literature, namely the CGLS detector~\cite{Yin2015}, the Neumann-series detector~\cite{Wu2014}, the OCD detector~\cite{Wu2016}, and the GS detector~\cite{WZXXY2016}. 
\tblref{tbl:complexity} shows the number of real-valued multiplications for $t_\text{max}$ iterations. We see that the complexity of TASER (for BPSK and QPSK) and the Neumann-series detector scales with $t_\text{max}U^3$, whereas TASER is slightly more complex; CGLS and GS both scale with $t_\text{max}U^2$, whereas GS is slightly more complex; OCD scales with $t_\text{max}BU$. Evidently, the near-ML performance of TASER comes at the cost of high computational complexity. In contrast, CGLS, OCD, and GS are rather inexpensive, but also perform poorly in square systems (see the $32\times32$ results in~\figref{fig:16x16_ver}). We finally note that TASER can be used for JED---the other (approximate) linear detectors cannot be used for this application.}
\begin{table}[tp]
\centering
{\color{black}\caption{Computational complexity of different data detection algorithms~for massive MIMO systems}
\label{tbl:complexity}
\begin{minipage}[c]{1\columnwidth}
\centering
\begin{tabular}{ll}
\toprule
{Algorithm} & Computational complexity\footnote{\rev{The complexity is measured by the number of real-valued multiplications for $t_\text{max}$ iterations. Complex-valued multiplications are assumed to require four real-valued multiplications. All  results ignore the preprocessing complexity.}} \tabularnewline
\midrule
BPSK TASER &  $t_\text{max}(\frac{1}{3}U^3+\frac{5}{2}U^2+\frac{37}{6}U+4)$   \tabularnewline
QPSK TASER & $t_\text{max}(\frac{8}{3}U^3+10U^2+\frac{37}{3}U+4)$ \tabularnewline
CGLS~\cite{Yin2015} & $(t_\text{max}+1)(4U^2+20U)$  \tabularnewline
Neumann~\cite{Wu2014} & $(t_\text{max}-1)2U^3+2U^2-2U$  \tabularnewline
OCD~\cite{Wu2016} & $t_\text{max}(8BU + 4U)$ \tabularnewline
GS~\cite{WZXXY2016} &  $t_\text{max}6U^2$  \tabularnewline  
\bottomrule
\end{tabular}
\end{minipage}}
\end{table} 
\subsection{FPGA Implementation Results}
\label{sec:FPGAdesign}
To demonstrate the effectiveness of TASER, we developed several FPGA designs for systolic array sizes of  $N=9$, $N=17$, $N=33$, and $N=65$. The FPGA designs were  implemented using Xilinx Vivado Design Suite and optimized for a Xilinx Virtex-7 XC7VX690T FPGA. The associated implementation results are shown in \tblref{tbl:implresultstaser}. 
As expected, the resource utilization increases quadratically with the array size $N$. 
For the $N=9$ and $N=17$ arrays, the critical path is in the PEs' MAC unit; for the $N=33$ and $N=65$ arrays, the critical path is in the row broadcast multiplexers, which limits the throughput of the $N=65$ array. 
\begin{table}[tp]
 \begin{minipage}[c]{1\columnwidth}
    \centering
    \caption{Implementation results on a Xilinx Virtex-7 XC7VX690T~FPGA~for different TASER array sizes}
       \label{tbl:implresultstaser}
      \begin{tabular}{lllll}
  \toprule
  {Array size} & ${N=9}$ & ${N=17}$ & ${N=33}$ & ${N=65}$ \tabularnewline
  {BPSK users / time slots} & ${8}$ & ${16}$ & ${32}$ & ${64}$ \tabularnewline
  {QPSK users / time slots} & ${4}$ & ${8}$ & ${16}$ & ${32}$ \tabularnewline
  \midrule
  {Slices} & 1\,467 & 4\,350 & 13\,787 & 60\,737 \tabularnewline
  {LUTs} & 4\,790 & 13\,779 & 43\,331 & 149\,942 \tabularnewline
  {FFs} & 2\,108 & 6\,857 & 24\,429 & 91\,829 \tabularnewline
  {DSP48s} & 52 & 168 & 592 & 2\,208 \tabularnewline
  {Max.\ clock freq.\ [MHz]} & 232 & 225 & 208 & 111 \tabularnewline
  {Min.\ latency [clock cycles]} & 16 & 24 & 40 & 72 \tabularnewline
  {Max.\ throughput [Mb/s]} & 116 &  150 &  166 & 98 \tabularnewline
  {Power estimate\footnote{Statistical power estimation at max.\ clock freq. and 1.0\,V supply voltage.} [W]} & 0.6 &  1.3 &  3.6 & 7.3 \tabularnewline
  \bottomrule
  \end{tabular}
  \end{minipage}
\end{table}
\rev{In \tblref{tbl:implcomp}, we compare TASER to the few existing large MIMO data detector designs, namely CGLS detector~\cite{Yin2015}, the Neumann-series detector~\cite{Wu2014}, the OCD detector~\cite{Wu2016}, and the GS detector~\cite{WZXXY2016}. All of these detectors have been implemented on the same FPGA and for a $128\times8$ large-MIMO system.}
    \begin{table*}
      \begin{minipage}[c]{1\textwidth}
		  \centering
    \caption{Comparison of $128\times8$ large-MIMO detectors on a Xilinx Virtex-7 XC7VX690T FPGA }
       \label{tbl:implcomp}
    \vspace{-1mm}
 \scalebox{1.0}{
  \begin{tabular}{lllllll}
  \toprule
  {Detection algorithm} & TASER & TASER & CGLS~\cite{Yin2015} & Neumann~\cite{Wu2014} & OCD~\cite{Wu2016} & GS~\cite{WZXXY2016} \tabularnewline
  {Error-rate performance}    & Near-ML & Near-ML & Near-MMSE & Near-MMSE & Near-MMSE & Near-MMSE\tabularnewline 
  {Modulation scheme} & BPSK & QPSK & 64-QAM & 64-QAM & 64-QAM & 64-QAM \tabularnewline
  {Preprocessing} & Not included & Not included  & Included & Included & Included & Included \tabularnewline
  {Iterations \rev{$t_\text{max}$}} & 3 & 3 & 3 & 3 & 3 & 1\tabularnewline
  \midrule
  {Slices} & 1\,467 (1.35\,\%) & 4\,350 (4.02\,\%) & 1\,094 (1\,\%) & 48\,244 (44.6\,\%) & 13\,447 (12.4\,\%) & n.a. \tabularnewline
  {LUTs} & 4\,790 (1.11\,\%) & 13\,779 (3.18\,\%) & 3\,324 (0.76\,\%) & 148\,797 (34.3\,\%) & 23\,955 (5.53\,\%) & 18\,976 (4.3\,\%) \tabularnewline
  {FFs} & 2\,108 (0.24\,\%) & 6\,857 (0.79\,\%) & 3\,878 (0.44\,\%) &   161\,934 (18.7\,\%) & 61\,335 (7.08\,\%) & 15\,864 (1.8\,\%)  \tabularnewline
  {DSP48s} & 52 (1.44\,\%) & 168 (4.67\,\%) & 33 (0.9\,\%)& 1\,016  (28.3\,\%) & 771 (21.5\,\%) & 232 (6.3\,\%)\tabularnewline
  {BRAM18s} & 0 (0\,\%) & 0 (0\,\%) & 1 (0.03\,\%) & 32\footnote{\label{footnote:BRAM}These designs use BRAM36s, which are equal to two BRAM18s.} (1.08\,\%) & 1 (0.03\,\%) & 12\textsuperscript{\ref{footnote:BRAM}} (0.41\,\%) \tabularnewline  
    \midrule
  {Clock frequency [MHz]} & 232 & 225 & 412 & 317 & 263 & 309 \tabularnewline
  {Latency [clock cycles]} & 48 & 72 & 951 & 196 & 795 & n.a.  \tabularnewline
  {Throughput [Mb/s]} & 38 &  50 & 20 &  621  & 379 & 48 \tabularnewline
  \midrule
  {Throughput/LUTs}  & {7\,933} & {3\,629} & {6\,017} & {4\,173} &  {15\,821}  & 2\,530 \tabularnewline
  \bottomrule
  \end{tabular}
  } % end of scalebox
  \end{minipage}
  \end{table*}
TASER achieves comparable throughput to the CGLS and GS designs and significantly lower latency than the Neumann-series and CD detectors. In terms of the hardware-efficiency (measured in terms of throughput per FPGA LUTs), our design performs similarly to CGLS, Neumann and GS, and inferior to the CD design. 
\rev{For the $128\times8$ massive MIMO system, all detectors achieve near-ML performance.  However, when considering the $32\times32$ large MIMO system (see Figures \ref{fig:16x16bpsk_ver} and \ref{fig:16x16qpsk_ver}), TASER  significantly outperforms the error-rate performance of all these reference designs.}
\rev{We conclude by noting that the CGLS, Neumann, OCD, and GS detectors are able to support 64-QAM, whereas TASER is limited to either BPSK or QPSK. This limitation negatively affects the throughput and hardware-efficiency of TASER, as the throughput of the other (approximate)  methods scales linearly in the number of bits per symbol---the provided throughput and hardware-efficiency results favor the CGLS, Neumann, OCD, and GS detectors.}
\subsection{ASIC Implementation Results}
\label{sec:ASICdesign}

 \begin{table}[tp]
 \vspace{-0.1cm}
    \caption{ASIC implementation results for different TASER array sizes}
       \label{tbl:implresultstaserASIC}
      \begin{minipage}[c]{1\columnwidth}
          \centering
  \begin{tabular}{llll}
  \toprule
  {Array size} & ${N=9}$ & ${N=17}$ & ${N=33}$ \tabularnewline
  {BPSK users / time slots} & ${8}$ & ${16}$ & ${32}$ \tabularnewline
  {QPSK users / time slots} & ${4}$ & ${8}$ & ${16}$ \tabularnewline
  \midrule
  {Core area [$\mu\mathrm{m}^2$]} & 149\,738 & 482\,677 & 1\,382\,318 \tabularnewline
  {Core density [$\%$]} & 69.86 & 68.89 & 72.89 \tabularnewline
  {Cell area [GE\footnote{One gate equivalent (GE) refers to the area of a unit-sized NAND2 gate.}]} & 148\,264 & 471\,238 & 1\,427\,962 \tabularnewline    
  {Max.\ clock freq.\ [MHz]} & 598 & 560 & 454 \tabularnewline
  {Min.\ latency [clock cycles]} & 16 & 24 & 40 \tabularnewline
  {Max.\ throughput [Mb/s]} & 298 &  374 &  363 \tabularnewline
  {Power estimate\footnote{Post-place-and-route power estimation at max.\ clock \rev{freq.} and $1.1$\,V.} [mW]} & 41 & 87 & 216 \tabularnewline
  \bottomrule
  \end{tabular}
  \end{minipage}
  \end{table} 
  
    \begin{table*}[tp]
	\vspace{-0.1cm}
    \centering
    \caption{Area breakdown for different TASER ASIC array sizes in gate equivalents (GEs)}
     \label{tbl:areataserASIC}
  \begin{tabular}{lllllll}
  \toprule
  {Array size} & \multicolumn{2}{c}{$N=9$} & \multicolumn{2}{c}{$N=17$} & \multicolumn{2}{c}{$N=33$}\tabularnewline
  \midrule
  {Element} & {Unit area} & {Total area} & {Unit area} & {Total area} & {Unit area} & {Total area} \tabularnewline
  \midrule
  {PEs } & 2\,391 (1.6\,\%) & 105\,198 (70.9\,\%) & 2\,404 (0.5\,\%) & 365\,352 (77.5\,\%) & 2\,086 (0.1\,\%) & 1\,168\,254 (81.8\,\%) \tabularnewline  
  {Scale units} & 6\,485 (4.4\,\%) & 25\,941 (17.5\,\%) & 6\,315 (1.3\,\%) & 50\,521 (10.7\,\%) & 5\,945 (0.4\,\%) & 95\,125 (6.6\,\%) \tabularnewline  
  {$\widehat{T}_{k,j}$ memories} & 734 (0.5\,\%) & 5\,873 (4.0\,\%) & 1\,451 (0.3\,\%) & 23\,220 (4.9\,\%) & 2\,888 (0.2\,\%) & 92\,426 (6.5\,\%) \tabularnewline  
  {Control unit} & 459 (0.3\,\%) & 459 (0.3\,\%) & 728 (0.2\,\%) & 728 (0.2\,\%) & 1\,259 (0.1\,\%) & 1\,259 (0.1\,\%) \tabularnewline  
  {Miscellaneous} & -- & 10\,793 (7.3\,\%) & -- & 31\,417 (6.7\,\%) & -- & 70\,898 (5.0\,\%) \tabularnewline  
  \bottomrule
  \end{tabular}

  \end{table*}

We also developed reference ASIC designs for systolic array sizes of $N=9$, $N=17$ and $N=33$. The ASIC designs were implemented using Synopsys DC and IC Compiler and optimized for a TSMC 40\,nm CMOS process. The associated implementation results are shown in \tblref{tbl:implresultstaserASIC}.  
As for our FPGA designs, the silicon area increases quadratically with the array size $N$. This can be verified both visually in \figref{fig:backend_color} as well as numerically in \tblref{tbl:areataserASIC}, where we see that the unit areas of each PE and scale unit remain nearly constant, while the total area of the PEs increases with $N^2$. As expected, the \rev{unit} area of the $\widehat{T}_{k,j}$ memories increases with $N$, as each one of these memories contains a column of an $N\times N$ matrix.
The critical paths for the $N=9$, $N=17$, and $N=33$ arrays are within the PE's MAC unit, the inverse square root LUT, and the row broadcasting bus, respectively.

\begin{figure}[tp]
\begin{tabular}{cc}
	\begin{adjustbox}{valign=t}
		\begin{tabular}{@{}c@{}}
		\subfigure[N=9]{\includegraphics[width=0.192 \columnwidth]{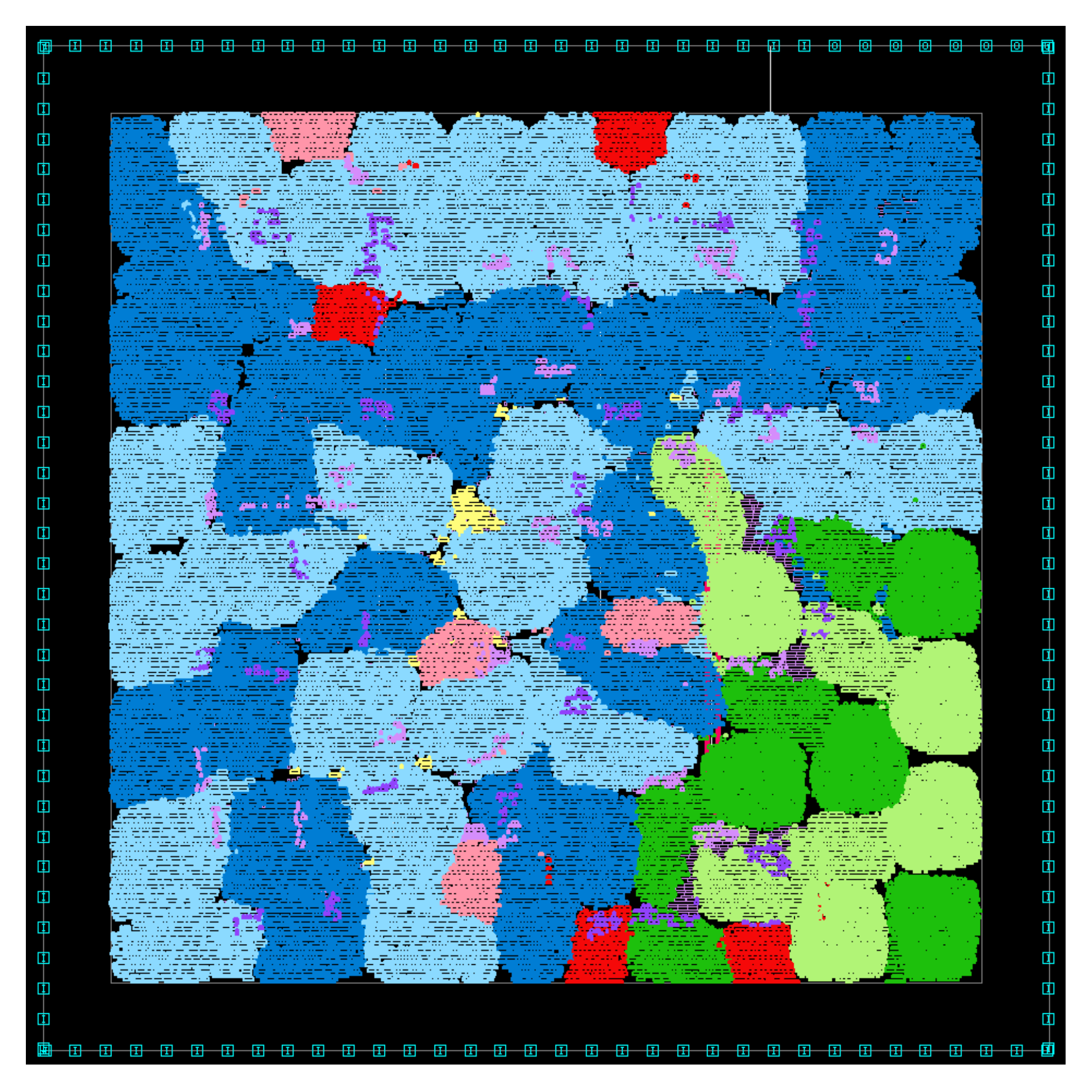}\label{fig:backend_N9_color}}
		\\
		\subfigure[N=17]{\includegraphics[width=0.34 \columnwidth]{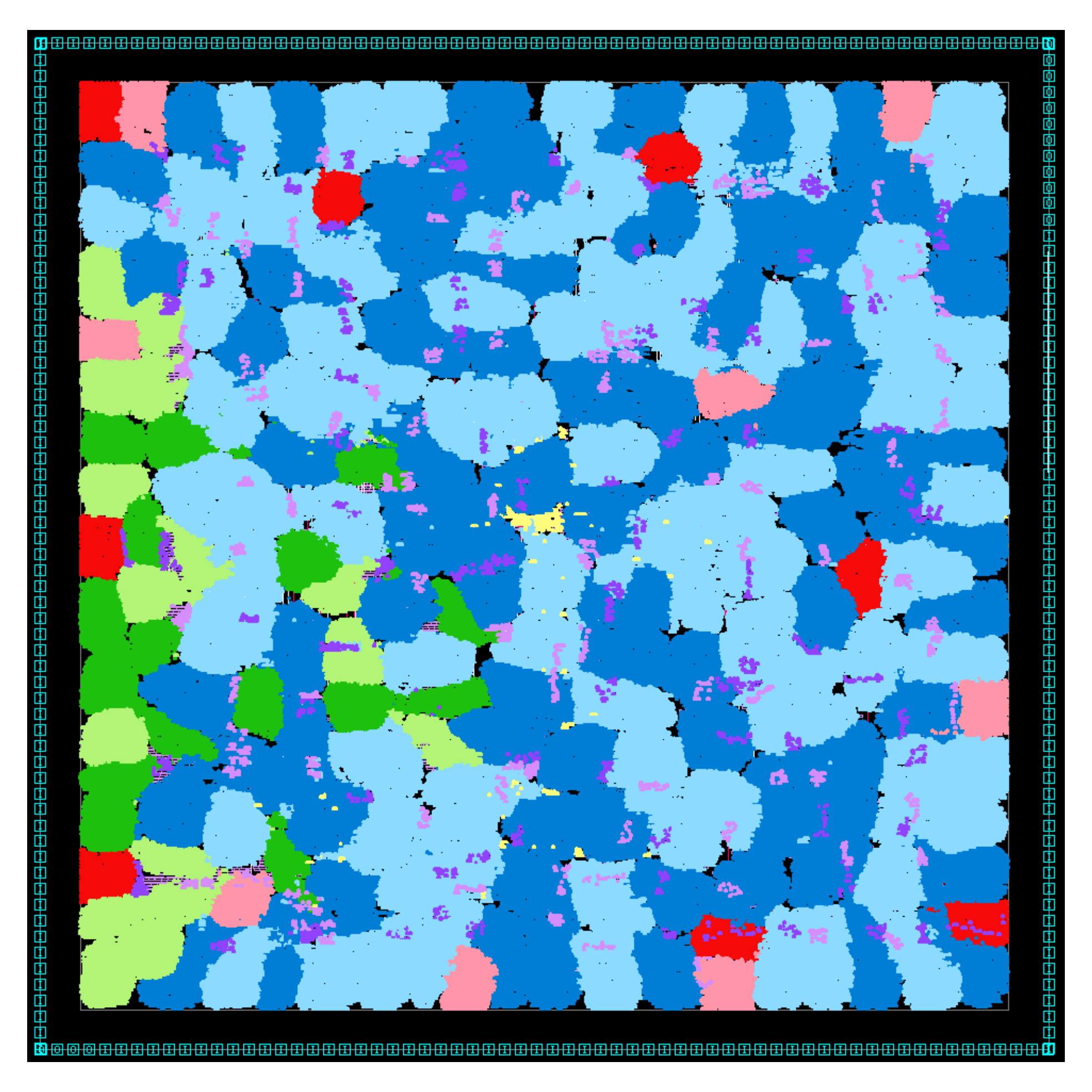}\label{fig:backend_N17_color}}
		\end{tabular}
	\end{adjustbox}
	&
	\begin{adjustbox}{valign=t}
     	\subfigure[N=33]{\includegraphics[width=0.58\columnwidth]{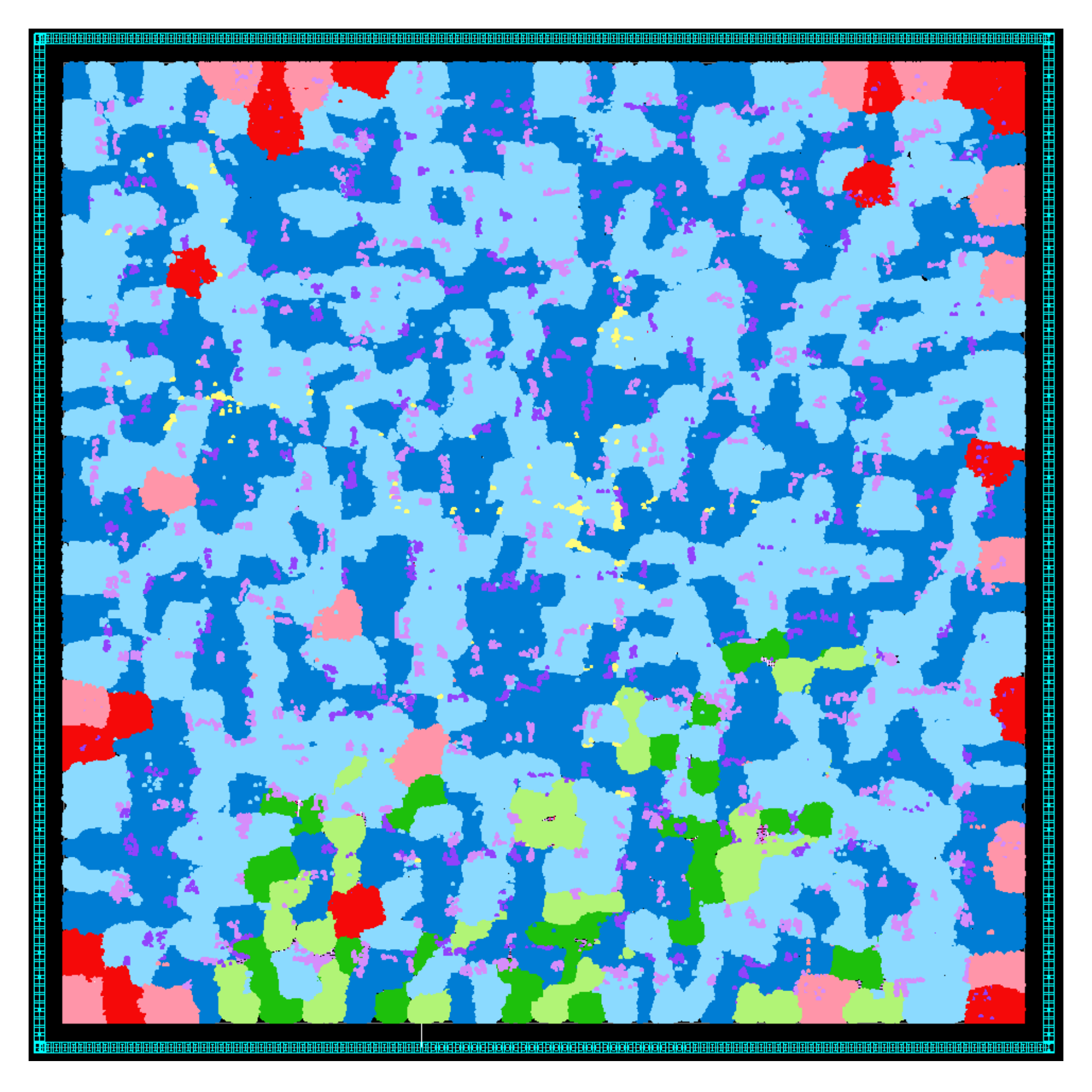}\label{fig:backend_N33_color}}		
	\end{adjustbox}
\end{tabular}
\caption{Layout  of the TASER ASIC designs for $N=9$, $N=17$ and $N=33$ array sizes. The different modules of the design were colored in the following way: PEs are colored in blue, memories in red, the scale units in green, the busses of the RBUs and CBUs in purple, and the control unit in yellow. Light and dark versions of the same color are alternated according to the order in which the modules appear in the hardware description code.}
\label{fig:backend_color}  
\end{figure}

 \begin{table}
  \begin{minipage}[c]{1\columnwidth}
            \centering
    \caption{Comparison of data detection ASICs for \mbox{$128$ BS antennas, $8$ users large-MIMO systems}   \label{tbl:implcompASIC}}
    \vspace{-1mm}
    \scalebox{0.95}{
  \begin{tabular}{llllll}
  \toprule
  {Detection algorithm} & TASER & TASER & Neumann~\cite{YWWDCS14b} \tabularnewline
  {Error-rate performance}    & Near-ML & Near-ML & Near-MMSE \tabularnewline 
  {Modulation scheme} & BPSK & QPSK & 64-QAM \tabularnewline
  {Preprocessing} & Not included & Not included & Included \tabularnewline
  {Iterations} & 3 & 3 & 3  \tabularnewline
  \midrule
  {CMOS technology [nm]} & 40 & 40 & 45 \tabularnewline
  {Supply voltage [V]} & 1.1 & 1.1 & 0.81 \tabularnewline
  \midrule
  {Clock freq.\ [MHz]} & 598 & 560 & 1\,000 (1\,125\footnote{\label{footnote:scaling}Technology scaling to 40\,nm and 1.1\,V assuming: $A \thicksim 1/\ell^2$, $t_{pd} \thicksim 1/\ell$, and $P_{dyn} \thicksim 1/(V_\ell^2 \ell)$ \cite{Razavi02}.}) \tabularnewline
  {Throughput [Mb/s]} & 99 & 125 & 1\,800 (2\,025\textsuperscript{\ref{footnote:scaling}}) \tabularnewline
  {Core area [mm$^2$]} & 0.150 & 0.483 & 11.1 (8.77\textsuperscript{\ref{footnote:scaling}})\tabularnewline
  {Core density [\%]} & 69.86 & 68.89 & 73.00\tabularnewline
  {Cell area\footnote{Excluding the gate count of memories.} [kGE]} & 142.4 & 448.0 & 12\,600 \tabularnewline
  {Power\footnote{At maximum clock frequency and given supply voltage.} [mW]} & 41.25 & 87.10 & 8\,000 (13\,114\textsuperscript{\ref{footnote:scaling}}) \tabularnewline
  \midrule
  {Throughput/cell}  & \multirow{2}{*}{695} & \multirow{2}{*}{279} & \multirow{2}{*}{161} \tabularnewline
  {area\textsuperscript{\ref{footnote:scaling}} [b/(s$\times$GE)]} & & & \tabularnewline
  {Energy/bit\textsuperscript{\ref{footnote:scaling}} [pJ/b]}  & 417 & 697 & 6\,476 \tabularnewline    
  \bottomrule
  \end{tabular}
  }
  %\normalsize
%  \vspace{-0.7cm}
  \end{minipage}

  \end{table}

In \tblref{tbl:implcompASIC}, we compare our TASER ASIC implementation to the Neumann-series detector in \cite{YWWDCS14b}, which is---to the best of our knowledge---the only ASIC design for massive MU-MIMO systems. While the latter offers a significantly higher throughput than our design, TASER's reduced area and power consumption result in superior hardware-efficiency (measured in throughput per cell area) and power-efficiency (measured in energy per bit). Furthermore, TASER enables near-ML performance for massive MU-MIMO systems where the number of users is in the same range as the number of BS antennas (see Figures \ref{fig:16x16bpsk_ver} and \ref{fig:16x16qpsk_ver}).
We note that the comparison presented on \tblref{tbl:implcompASIC} is not entirely fair. TASER does not include preprocessing circuitry, whereas the Neumann-series detector~\cite{YWWDCS14b} includes preprocessing circuitry and was optimized for wideband systems that use single-carrier frequency-division multiple access (SC-FDMA).

\rev{We finally note that there exists a plethora of data-detector ASICs for traditional, small-scale MIMO systems (see~\cite{wong2002vlsi,burg2005vlsi,jsac07,studer2010vlsi,5405138,5560294,5570931,liao20143,6725688} and the references therein). While most of these designs achieve near-ML performance and/or throughputs in the Gb/s regime in small-scale MIMO systems, their efficacy for large MIMO is unexplored---a corresponding algorithm and hardware-level comparison is left for future work.}

%%%%%%%%%%%%%%%%%%%%%%%%%%%%%%%%%%%%%%%%%%%%%%%%%%%%

\section{Conclusions}
\label{sec:conclusions}

We have proposed---to the best of our knowledge---the first data-detector implementation that uses semidefinite relaxation. 
\rev{Our novel algorithm, referred to as Triangular Approximate SEmidefinite Relaxation (TASER), is suitable for coherent data detection in massive MU-MIMO systems, as well as joint channel estimation and data detection (JED) in large SIMO systems.}
\rev{We have developed a corresponding systolic VLSI architecture and implemented FPGA and ASIC reference designs. Our results have shown that TASER achieves comparable hardware-efficiency as existing  massive MU-MIMO data detectors, while providing near-ML performance, even for systems where the number of users is comparable to the number of BS antennas.
Hence, for systems supporting a large number of low-rate users (e.g., $16$ users or more) where BPSK and QPSK transmission is sufficient, TASER provides a viable alternative to sub-optimal, linear data-detection methods, or optimal but computationally expensive non-linear methods.
We also note that TASER can be used in so-called overloaded systems, i.e., systems with more users than BS antennas---such a scenario may be of interest in large sensor networks or for the internet of things (IoT).}

There are many avenues for future work. Traditional SDR-based data detection only supports BPSK and QPSK transmission and hard-output \rev{data} detection. Extending TASER to support higher-order modulation schemes using the methods in~\cite{wiesel2005semidefinite, sidiropoulos2006} is the subject of ongoing research. 
Furthermore, developing efficient ways to compute soft-output values (in the form of log-likelihood ratio values) within TASER is left for future work.
Finally, SDR-based data detection for JED in MU-MIMO systems is an interesting open research topic. 

%%%%%%%%%%%%%%%%%%%%%%%%%%%%%%%%%%%%%%%%%%%%%%%%%%%%

\section*{\rev{Acknowledgments}}
\sloppy
\rev{The authors would like to thank Prof.\ C.\ Batten for his support with the ASIC design flow. C.~Studer would like to thank Prof.\ W.~Xu for insightful discussions on JED. O.~Casta\~neda would like to thank Cornell University's CienciAmerica Program that enabled this  research project. O. Casta\~neda and C. Studer were supported in part by Xilinx Inc., and by the US National Science Foundation (NSF) under grants ECCS-1408006 and CCF-1535897. T.~Goldstein was supported in part by the US NSF under grant CCF-1535902 and by the US Office of Naval Research under grant N00014-15-1-2676.}
\fussy

%%%%%%%%%%%%%%%%%%%%%%%%%%%%%%%%%%%%%%%%%%%%%%%%%%%%

\balance

%%%%%%%%%%%%%%%%%%%%%%%%%%%%%%%%%%%%%%%%%%%%%%%%%%%%

% that's all folks
\end{document}